%% file: planar.tex
\documentclass[10pt,screen,acmsmall]{acmart}

\setcopyright{none}
\copyrightyear{2022}

\acmYear{2022}

\acmJournal{PACMPL}
\acmVolume{1}
\acmNumber{X} 
\acmArticle{1}

\usepackage[T1]{fontenc}
\usepackage[utf8]{inputenc}
\usepackage{utf8-symbols}

\usepackage{microtype}
\usepackage{nccmath}
\usepackage{textcomp}
\usepackage{mathtools}
\usepackage{xparse}
\usepackage{bbold}
\usepackage{stmaryrd}
\usepackage{wrapfig}
\usepackage{multirow}
\usepackage{fancyvrb}
\usepackage{xcolor}
\usepackage{float}
\usepackage{graphicx}
\usepackage{amsmath}
\usepackage{amsthm}
\usepackage{dirtytalk}
\usepackage{epigraph}
\usepackage{tikz}
\usepackage{quiver}
\usepackage{subcaption}
\allowdisplaybreaks

\definecolor{darkgreen}{HTML}{014801}
\definecolor{darkblue}{HTML}{191970}
\definecolor{antiquefuchsia}{rgb}{0.57, 0.36, 0.51}
\definecolor{color1}{rgb}{0.94 0.93 0.96}
\definecolor{color2}{rgb}{0.74 0.74 0.86}
\definecolor{color3}{rgb}{0.46 0.42 0.69}
\definecolor{color4}{rgb}{0.53 0.34 0.65}

\usepackage[all,british]{foreign}

\usepackage{mathpartir}
\usepackage[nameinlink,capitalize,noabbrev]{cleveref}

\usepackage{thmtools}

\newenvironment{linkproof}[2][]{%
    \begin{proof}[#1]
        
    }{%
    \end{proof}
    
}
{%
        \hfill\href{\linkagdadefinition}{\(\diamond\)}
        \end{definition}
}
{%
        \hfill\href{\linkagdatheorem}{\(\diamond\)}
        \end{theorem}
}
{%
        \hfill\href{\linkagdalemma}{\(\diamond\)}
        \end{lemma}
}
\numberwithin{equation}{section}

\usetikzlibrary{decorations.markings}
\usetikzlibrary{decorations.pathreplacing}
\usetikzlibrary{matrix}
\usetikzlibrary{arrows}
\usetikzlibrary{chains}
\usetikzlibrary{positioning}
\usetikzlibrary{scopes}
\usetikzlibrary{decorations.pathmorphing}

\usepackage[textsize=footnotesize, obeyFinal]{todonotes}

\input{macros/names.tex}

\graphicspath{{./../reports/ipe-images/}}

\hypersetup{
  unicode=true,
  pdfkeywords={planarity, combinatorial maps, univalent mathematics, formalisation of mathematics},
  pdftitle={On Planarity of Graphs in Homotopy Type Theory},
  pdfauthor={J. Prieto-Cubides and H. R. Gylterud},
  colorlinks=true,
  linkcolor=black,
  citecolor=black,
  urlcolor=black,
  breaklinks=true
}

\usepackage[nameinlink,capitalize,noabbrev]{cleveref}
\Crefname{theorem}{Theorem}{Theorems}
\Crefname{lem}{Lemma}{Lemmas}
\Crefname{lem}{Corollary}{Corollaries}
\Crefname{listthm}{Theorem}{Theorems}
\Crefname{listlem}{Lemma}{Lemmas}
\crefname{axiom}{Axiom}{Axioms}
\Crefname{axiom}{Axiom}{Axioms}
\Crefname{fun}{Function}{Functions}
\Crefname{eq}{}{}
\Crefname{equiv}{Equivalence}{Equivalences}
\Crefname{equality}{Equality}{Equalities}
\Crefname{cond}{Condition}{Conditions}
\Crefname{calc}{Calculation}{Calculations}
\Crefname{type}{Type}{Types}
\creflabelformat{eq}{#2(#1)#3}
\creflabelformat{type}{#2(#1)#3}
\creflabelformat{equality}{#2(#1)#3}
\creflabelformat{calc}{#2(#1)#3}
\creflabelformat{equiv}{#2(#1)#3}
\Crefname{fig}{Figure}{Figures}
\Crefname{proof}{Proof}{Proofs}

\usepackage{datetime2} 

\usepackage{caption}
\usepackage{flushend}

\begin{document}

\title{On Planarity of Graphs in Homotopy Type Theory
}
\author{Jonathan Prieto-Cubides}
\orcid{1234-4564-1234-4565}
\affiliation{
  \department{Department of Informatics}
  \institution{University of Bergen}
          \postcode{5073}
      \country{Norway}
  }
\email{jonathan.cubides@uib.no}
\author{Håkon Robbestad Gylterud}
\orcid{1234-4564-1234-4565}
\affiliation{
  \department{Department of Informatics}
  \institution{University of Bergen}
            \country{Norway}
  }
\email{hakon.gylterud@uib.no}

\input{planar/abstract.md}



\maketitle

\hypertarget{introduction}{%
\section{Introduction}\label{introduction}}

Topological graph theory studies the embedding of graphs into surfaces
\citep{gross, surveytopgraph, Stahl1978} such as: the plane, the
sphere, the torus, etc. Even the simplest case, embedding graphs into
the plane, has inspired a lot of interesting characterisations and
mathematical results. Two such characterisations are Kuratowski's
theorem and the closely related Wagner's theorem
\citep{diestel, Rahman2017}. Both theorems characterise planarity by
excluding two sub graphs known as the forbidden minors, namely \(K_5\)
and \(K_{3,3}\). Other approaches refer to algebraic methods as
MacLane's Theorem \citep{maclane} and Schnyder's theorem
\citep[§3.3]{Baur2012}.

One of the most powerful tools in topological graph theory is the
\emph{combinatorial representation} of graph embeddings, called
rotation systems \citep{gross}. These representations encode what the
embedding looks like around each node -- characterising the embedding
up to isotopy. \Cref{sec:graph-embeddings} will give a more detailed
description of rotation systems. For now, it suffices to know that for
suitably general class of embeddings into closed surfaces -- namely,
the cellular ones -- the embedding is characterised by the cyclic
order of outgoing edges from each node as they lie around the node on
the surface.

Homotopy type theory (HoTT)\citep{hottbook} is a variation of
dependent type theory which emphasises the higher dimensional
structure of types: Equalities within a type are seen as paths, and
the type of all equalities between two elements -- the identity type
-- is thought of as a path space. In this way, HoTT takes seriously
the notion of proof-relevancy, and interesting questions arise when
considering what the equality between two proofs are.

The goal of this paper is to develop a proof-relevant notion of planar
graphs in homotopy type theory, based on rotation systems. In other
words, \emph{planarity is a structure on a graph}, not a mere
property. Intuitively, a proof that a graph is planar is an embedding
into the plane. The question is then, when are two such embeddings
equal? One good answer is that the proofs ought to be equal when the
embeddings are isotopic -- i.e.~can be deformed continuous to one
another without crossing edges. This will be the case for the notion
presented here. But in order to arrive at a type of graph embeddings
where the identity type corresponds to isotopy, a lot of care has to
be taken when defining embeddings, and planarity.

In short, a planar graph will be a graph with a combinatorial
embedding into the sphere and a fixed face where to puncture, as in
\Cref{fig:planar-embeddings}. The intuition is that an embedding into
the plane can be obtained from an embedding into the sphere by
puncturing the sphere at a point symbolising infinity (in any
direction) on the plane. Up to isotopy, the important data when
choosing a point of puncture is which face the points lies in.

\begin{figure} \centering
  \includegraphics[width=0.7\textwidth]{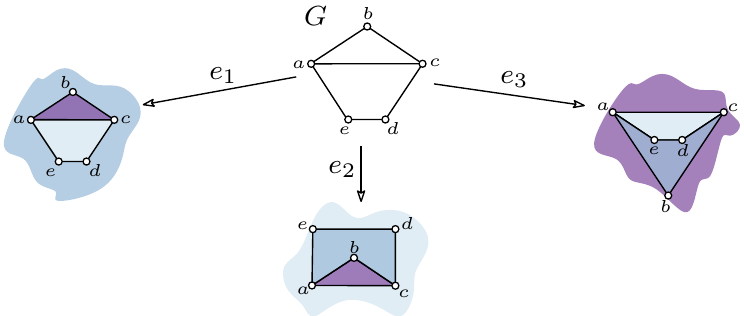}
  \caption{It is shown a graph along with three different planar
  embeddings, namely $e_1$, $e_2$, and $e_3$. We have shaded with
  different colours the three faces in each embedding.}
  \label{fig:planar-embeddings} \end{figure}

Our development here differs from other related works in the subject
(see \Cref{sec:related-work}), essentially by adopting the Voevodsky's
Univalence axiom (UA) present in HoTT. One consequence is that the
graphs maintain the structural identity principle, i.e.~isomorphic
graphs are equal and share the same structures and properties. Such a
correspondence turns to be crucial for formalising mathematics, as in
standard mathematical practice. For example, the identity type of a
graph helps us to understand its symmetries. Any automorphism of a
graph gives rise to an inhabitant of its identity type and vice versa.
One can then describe the group structure of the set of automorphisms
for a graph by studying its identity type \citep[§11]{Gross2018}, see
\Cref{sec:automorphisms}. We foresee an exciting opportunity to
combine ideas and prove results in combinatorics, graph theory, and
homotopy type theory.

\paragraph{Outline}

\Cref{sec:background} introduces the basic terminology and notation
used throughout the paper. In \Cref{sec:notions-of-graph-theory}, the
category of graphs is described, along with a few relevant examples.
In \Cref{sec:graph-embeddings}, we present different types for
graph-theoretic concepts, including the type of maps, faces, and
spherical maps, which allows us to define in HoTT the notion of planar
maps, and consequently, planar graphs in \Cref{sec:planar-embeddings}.
Additionally, to construct larger planar graphs, the planarity for
cyclic graphs and graph extensions are proved. \Cref{sec:related-work}
addresses the connection of this work with other developments. A few
concluding remarks and future work are discussed in
\Cref{sec:conclusions}.

\paragraph{Formalisation}

One exciting feature to work with systems such as HoTT is producing
machine-verified proofs \citep{harrison-notices}. To check the
correctness of this work, we use the proof assistant Agda
\citep{agda}, a computer system with support for dependent type
theories capable of managing the same abstraction level with which we
reason our mathematical theories on paper. Formal machine-verified
proofs can offer a window to new proofs and theorems
\citep{avigad-fvm}. Moreover, they also serve to find flaws and corner
cases that human reasoning might not see. We must therefore, pay
special attention to definitions and theorems, as they are the primary
input to these systems. The computer's formalisation process is an
exciting and challenging activity, full of details and technical
issues \citep{Gonthier2008, Appel1986}. We use Agda v\(2.6.2\) to
type-check the formalisation \citep{agdaformalisation} of the
essential parts of this paper. To be compatible with HoTT, we use the
flag \texttt{without-K} \citep{COCKX2016} and the flag
\texttt{exact-split} to ensure all clauses in a definition are
definitional equalities.

\section{Mathematical Foundation}\label{sec:background}

In this paper, we work with homotopy type theory
\citep{hottbook, Grayson2017}: a Martin-Löf intensional intuitionistic
type theory extended with the Univalence Axiom
\citep{Awodey2018, escardoUA}, proposed originally by Voevodsky
\citep{voevodsky2014equivalence}, and some higher inductive types
(HITs), such as propositional truncation. The presentation of our
constructions is informal in the style of the HoTT book
\citep{hottbook}. However, the essential constructions have been
verified in the proof assistant Agda \citep{agdaformalisation}.

HoTT emphasises the rôle of the identity type as a path-type
\citep{Awodey2012}. The intended interpretation is that elements,
\(a,a' : A\), are \emph{points} and that a witness of an equality
\(p : a = a'\) is a \emph{path} from \(a\) to \(a'\) in \(A\). Since
the identity type again is a type, we can iterate the process, which
gives each type the structure of an ∞-groupoid.

This may at first seem of little relevant when working with finite
combinatorics, as one would expect only types with trivial path-types
(sets) to show up in combinatorics. But we will see that types with
nontrivial path types do indeed arise naturally in combinatorics --
which is not surprise for someone familiar with the role of groups and
groupoids in this field, such as Joyal's work on combinatorial species
\citep{Baez2009, Yorgey}-- and that the paths in these types are often
various forms of permutations.

\hypertarget{notation}{%
\subsection{Notation}\label{notation}}

\begin{itemize}
\item
  Definitions are introduced by (\(:≡\)) while judgmental equalities
  use (\(≡\)).
\item
  The type \(\UU\) is an \emph{univalent} universe. The notation
  \(A : \UU\) indicates that \(A\) is a type. To state that \(a\) is
  of type \(A\) we write \(a : A\).
\item
  The equality sign of the identity type of \(A\) is denoted by
  (\(=_{A}\)). If the type \(A\) can be inferred from the context, we
  simply write \((=)\). The equalities between \(x,y: A\) are of type
  \(x =_{A} y\).
\item
  The type of non-dependent function between \(A\) and \(B\) is
  denoted by \(A \to B\).
\item
  Type equivalences are denoted by (\(≃\)). The canonical map for
  types is \(\mathsf{idtoequiv} : A = B → A ≃ B\) and its inverse
  function is called \(\mathsf{ua}\). Given \(e : A ≃ B\), the
  underlying function of type \(A \to B\) is denoted by \(e\) while
  \(\mathsf{ua}(e)\) is denoted by \(\overline{e}\). The coercion
  along \(p : A = B\) is the function \(\mathsf{coe}\) of type
  \(A \to B\).
\item
  The point-wise equality for functions (also known as
  \emph{homotopy}) is denoted by (\(\sim\)). The function
  \(\mathsf{happly}\) is of type \(f = g → f ∼ g\) and its inverse
  function is called \(\mathsf{funext}\).
\item
  The coproduct of two types \(A\) and \(B\) is denoted by \(A + B\).
  The corresponding data constructors are the functions
  \(\mathsf{inl} : A \to A + B\) and \(\mathsf{inr}: B \to A+B\).
\item
  Dependent sum type (Σ-type) is denoted by \(Σ_{x:A} B(x)\) while
  dependent product type (Π-type) is denoted by \(Π_{x:A}B(x)\).
\item
  The empty type and unit type are denoted by \(\mathbb{0}\) and
  \(\mathbb{1}\), respectively.
\item
  The type \(x \neq y\) denotes the function type
  \((x = y) \to \mathbb{0}\).
\item
  Natural numbers are of type \(\N\). The variable \(n\) is of type
  \(\N\), unless stated otherwise.
\item
  The type with \(n\) points is denoted by \([n]\).
\item
  The universe \(\UU\) is closed under the type formers considered
  above.
\item
  The function transport/substitution is denoted by \(\mathsf{tr}\).
  We denote by \(\mathsf{tr}_2\) the function of type
  \(\prod_{p : a₁ = a₂}\,(\tr{B}{p}{b_1} = b₂) → C(a₁, b₁) → C(a₂,b₂)\),
  where \(A : \UU\), \(B : A →\UU\), \(a₁,a₂ : A\), \(b₁ : B a₁\),
  \(b₂ : B a₂\), and \(C~:~\prod_{x :A}~(B x →\UU)\).
\end{itemize}

For the sake of readability in the upcoming section, we will use
variables \(A,B\) and \(X\) to denote types, unless stated otherwise.

\hypertarget{homotopy-levels}{%
\subsection{Homotopy Levels}\label{homotopy-levels}}

The following establishes a level hierarchy for types with respect to
the nontrivial homotopy structure of the identity type. The first four
homotopy levels are of special interest for this work. These four
levels are enough for expressing the mathematical objects we want to
construct.

\begin{definition}\label{def:n-types} Let $n$ be an integer such that
$n \geq -2$. One states that a type $A$ is an $n$\emph{-type} and that
it has homotopy level $n$ if the proposition
$\mathsf{is\mbox{-}level}(n,A)$ holds.
\begin{align*}
\mathsf{is\mbox{-}level}(-2, A)  &:≡ \sum_{(c:A)} \prod_{(x:A)} (c = x), \\
\mathsf{is\mbox{-}level}(n+1, A) &:≡ \prod_{(x,y:A)}\,\mathsf{is\mbox{-}level}(n,x = y).
\end{align*}
\end{definition}

For convenience of the presentation, one states that a type of
homotopy level \((-2)\) is a contractible type, a level \((-1)\) is a
proposition, a level \((0)\) is a set, and finally, the level \((1)\)
is a \(1\)-groupoid. We also define
\(\mathsf{isProp}(A):\equiv \mathsf{is\mbox{-}level}(-1, A)\) and
\(\mathsf{isSet}(A):\equiv \mathsf{is\mbox{-}level}(0, A)\). Types
that are propositions are of type \(\mathsf{hProp}\) and similarly
with the other levels. Additionally, it is possible to have an
\(n\)-type out of any type \(A\) for \(n ≥ −2\). This can be done
using the construction of a higher inductive type called
\(n\)-truncation \citep[§7.3]{hottbook} denoted by \(\|A\|_n\). The
case for \((-1)\)-truncation is called \emph{propositional truncation}
(or \emph{reflection}), and is often simply denoted by \(\|A\|\).

\begin{definition}\label{def:propositional-truncation}
\emph{Propositional truncation} of a type $A$ denoted by $\| A
\|_{-1}$ is the \emph{universal solution} to the problem of mapping
$A$ to a proposition $P$. The elimination principle of this
construction gives rise a map of type $\| A \| \to P$ which requires a
map $f : A → P$ and a proof that $P$ is a proposition.
\end{definition}

Propositional truncation allows us to model the \emph{mere} existence
of inhabitants of type \(A\). We state that \(x\) is \emph{merely}
equal to \(y\) when \(\|x = y\|\) for \(x,y : A\). Then, we can
express in HoTT by means of propositional truncation: logical
conjunction\footnote{\((P ∨ Q) :≡ \| P + Q \|\).}, the
disjunction\footnote{\((P ∧ Q) :≡ \| P × Q \|\).}, and existential
\footnote{\((\exists (x:A) P(x)) :≡ \| Σ_{x :A} Px \|\).}
quantification.

\begin{definition}\label{def:connected-component} Given $x:A$, the
\emph{connected component} of $x$ in $A$ is the collection of all
$y:A$ that are merely equal to $x$, i.e. $Σ_{y:A} \| y = x \|$. If $\|
x = y\|$, one says that $x$ is connected to $y$.
\end{definition}

\begin{definition}
The type $A$ is called \emph{connected} if $\|A\|$ holds and also, if
all $x:A$ belong to the same connected component.
\end{definition}

\begin{lemma}\label{lem:connected-share-same-predicates} Terms in the
same connected component share the same propositional properties. If
$P : A \to \mathsf{hProp}$ and $x, y : A$ are connected in $A$ then
one gets the equivalence $P(x) ≃ P(y)$.
\end{lemma}

\hypertarget{finite-types}{%
\subsection{Finite Types}\label{finite-types}}

The finiteness of a type \(A\) is the existence of a bijection between
\(A\) and the type \([n]\) for some \(n:\mathbb{N}\). However, this
description is not a structure on \(A\), providing it with a specific
equivalence \(A ≃ [n]\), but rather a property --- a mere proposition.
This ensures that the identity type on the total type of finite types
is free to permute the elements, without having to respect a chosen
equivalence.

\begin{definition}\label{def:finite-type}
Given $X : \UU$, let $\mathsf{isFinite}(X) : \UU$ be given by
\begin{equation*}
\mathsf{isFinite}(X) :≡ \sum_{(n:\N)} \left\| X \simeq [n] \right\|.
\end{equation*}
\end{definition}

\begin{lemma}\label{lem:is-finite-is-prop}
The type $\mathsf{isFinite}(X)$ is
a proposition.
\end{lemma}

\begin{proof} Let $(n, p), (m, q) : \mathsf{isFinite}(X)$, which we
want to prove equal. Since $p$ and $q$ are elements of a family of
propositions, it is sufficient to show that $n=m$. This equation is a
proposition, so we can apply the truncation-elimination principle to
get $X ≃ [n]$ and $X ≃ [m]$. Thus, from $[n] ≃ [m]$ follows that
$n = m$ --- by a well-known result on finite sets. \end{proof}

A type \(X\) is \emph{finite} if \(\mathsf{isFinite}(X)\) holds. The
natural number \(n\) is referred as the cardinal number of \(X\). Now,
since \(\mathsf{isFinite}(X)\) is a proposition, the total type of
finite types, \(\sum_{X :\UU} \mathsf{isFinite}(X)\), has permutations
as its identity type. This shows that \Cref{def:finite-type} is
equivalent to the type \(∃_{n:N} (X = [n])\). However, the former
definition allows us to easily obtain \(n\) by projecting on the first
coordinate. We find this distinction more practical than the latter
for certain proofs, as in \Cref{lem:cyclic-equation}. Furthermore, one
can show that any property on \([n]\), for example, \say{being a set}
and \say{having
decidable equality} can be transported to any finite type.

\subsection{Cyclic Types}\label{sec:cycle-types}

The notion of being a cycle for a type is not a property but a
structure. The description of such a structure can be addressed in
several ways. For example, one can state that a type is cyclic if it
has a cycle order, i.e.~there exists a (ternary) relation on the type
for which the axioms of a cyclic ordering hold true \citep{Haar2014}.
In our approach, \Cref{def:cyclic-type}, we establish that being
cyclic for a type is a structure given by preserving the structure of
cyclic subgroups of permutations on \([n]\).

\begin{wrapfigure}[6]{R}[0pt]{0pt}
  \centering
  \begin{tikzcd}
  n-1 \arrow[d, "\mathsf{pred}^{n-i-1}", dotted, maps to] & 0 \arrow[l, "\mathsf{pred}"', maps to] & 1 \arrow[l, "\mathsf{pred}"', maps to]               \\[2pt]
  i \arrow[rr, "\mathsf{pred}"', maps to]                 &                                        & i-1 \arrow[u, "\mathsf{pred}^{i-2}"', dotted, maps to]
  \end{tikzcd}
  \end{wrapfigure}

As with natural numbers, we can define counterpart functions for the
successor and the predecessor in \([n]\) if \(n \geq 1\). The
predecessor function \(\mathsf{pred}\), of type \([n] → [n]\), can be
defined as the mapping: \(0↦ (n-1)\) and \((n+1) ↦n\). The
corresponding \(\mathsf{succ}\) function is the inverse of
\(\mathsf{pred}\), and they are therefore both equivalences. The
\(\mathsf{pred}\) function generates a cyclic subgroup (of order
\(n\)) of the group of permutations on \([n]\). An equivalent cyclic
subgroup can be defined by means of the \(\mathsf{succ}\) function;
however, \(\mathsf{pred}\) is more straightforward to define than
\(\mathsf{succ}\). Now, we want to mirror the structure of \([n]\)
given by \(\mathsf{pred}\) for any finite type \(A\) along with an
endomap \(\varphi : A → A\). This can be done by establishing a
structure-preserving map between \((A, \varphi)\) and
\(([n], \mathsf{pred})\) in the category of endomaps of sets. This is
the idea behind \Cref{def:cyclic-type} and the fact that the condition
\emph{structure-preserving} can be attained.

\begin{definition}\label{def:cyclic-type}
The type of \emph{cyclic structures} on a type $A$ is given by $\mathsf{Cyclic}(A)$.
\begin{equation*}
\mathsf{Cyclic}(A) :≡ \sum_{(\varphi : A → A)}\sum_{(n : \N)} \| ∑_{(e : A ≃ [n])} e ∘ \varphi = \mathsf{pred} ∘ e \| .
\end{equation*}
\end{definition}

A cyclic structure on a type \(A\) is denoted by a triple
\(⟨A, f, n⟩\) where \((f, n, \mbox{-}) : \mathsf{Cyclic}(A)\). Given
such a triple, we can refer to \(A\) as an \(n\)-\emph{cyclic type}.
By \Cref{lem:finite-type-prop1,lem:finite-type-prop2}, one can prove
that the \(n\)-cyclic type \(A\) is not only a finite set, but the
function \(f\) is a bijection. To avoid any confusions, we denote
\(f\) by \(\mathsf{pred}_{A}\) and the inverse by
\(\mathsf{suc}_{A}\). One may also drop the previous sub-indices. To
define such bijections as products of cycles, we use the same notation
from group theory. For example, the permutation denoted by \((a)(bc)\)
permutes \(b\) and \(c\) and fixes \(a\).

\begin{lemma}\label{lem:finite-type-prop2}
Let $P$ be a family of propositions of type $\prod_{X:\UU} (X \to X) \to \mathsf{hProp}$
and an $n$-cyclic structure $⟨A, f, n⟩$.
If $P([n],\mathsf{pred})$, then $P(A , f)$.
\end{lemma}
\begin{proof}
It follows from \Cref{lem:connected-share-same-predicates}. Note that
being cyclic for a type is equivalent to saying $(A,f)$ and $([n],
\mathsf{pred})$ are connected in $(X:\UU) \times (X \to X)$.
\end{proof}

\begin{lemma}\label{lem:finite-type-prop1}
Let $P$ be a family of propositions of type $\UU \to \mathsf{hProp}$ and an
$n$-cyclic structure $⟨A, f, n⟩$. If $P([n])$, then $P(A)$.
\end{lemma}

In any finite type, every element is searchable. In particular, given
an \(n\)-cyclic type \(⟨A, f, n⟩\), one can search any element by
iterating at most \(n\) times, the function \(f\) on any other
element.

\begin{lemma}\label{lem:cyclic-equation} If $A$ is an $n$-cyclic type,
then there exists a unique number $k$ with $k<n$ such that
$\mathsf{pred}_{A}^k(a) = b$ for all $a,b:A$.
\end{lemma}

The total type, \(∑_{A:\UU} \mathsf{Cyclic}(A)\), is the classifying
type \citep[§4.6-7]{symmetrybook} of finite cyclic groups. In the
remaining of this section we compute the identity type between two
finite cyclic types that we use, for example, in
\Cref{ex:bouquet-maps} to enumerate the maps of the bouquet graph
\(B_2\).

\begin{lemma}\label{lem:cyclic-is-set}
  The type $\mathsf{Cyclic}(A)$ is a set.
\end{lemma}

\begin{lemma}\label{lem:cyclic-identity-type}
  Given $\mathcal{A}$ and $\mathcal{B}$ defined
  by $\langle A , f , n \rangle$ and $\langle B , g , m \rangle$,
  $$(\mathcal{A} = \mathcal{B}) \simeq  ∑_{(α:A = B)} (\coe{α} ∘ f = g ∘ \coe{α}).$$
\end{lemma}
\begin{proof}
We show the equivalence by \Cref{comp-identity-fin}. In
\Cref{comp-identity-fin-1}, we expand the cycle type definition for
$\mathcal{A}$ and $\mathcal{B}$. \Cref{comp-identity-fin-2} follows
from the characterisation of the identity type between pairs in a
$\Sigma$-type \cite[§3.7]{hottbook}. Note that in
\Cref{comp-identity-fin-2}, the type $(n = m) \times (p = q)$ is a
contractible type, i.e. equivalent to $\mathbb{1}$, the one-point
type. We can then simplify the inner $\Sigma$-type to its base in
\Cref{comp-identity-fin-2a} to obtain by the equivalence $\Sigma_{x:A}
\mathbb{1} ≃ A$, \Cref{comp-identity-fin-3}. Finally,
\Cref{comp-identity-fin-4} is a consequence of transporting functions
along the equality $\alpha$. The conclusion is that the identity type
$\mathcal{A}=\mathcal{B}$ is equivalent to the type of equalities $α:A
= B$ along with a proof that the structure of $f$ is preserved in the
structure of $g$. 

\begin{subequations}\label[calc]{comp-identity-fin}
\begin{align}
\label[equiv]{comp-identity-fin-1}
(\mathcal{A} = \mathcal{B})  &≡ ((A, (f, n, p )) = (B , (g, m, q))) \\
\label[equiv]{comp-identity-fin-2}
&≃ ∑_{(α:A = B)} ∑_{(β : \tr{\,\lambda X.X → X}{α}{f}=g)} (n = m) \times (p = q)\\
\label[equiv]{comp-identity-fin-2a}
&≃ ∑_{(α:A = B)} ∑_{(β : \tr{\,\lambda X.X → X}{α}{f}=g)} \mathbb{1}\\
\label[equiv]{comp-identity-fin-3}
&≃ ∑_{(α:A = B)} (\tr{\,\lambda X.X → X}{α}{f}=g) \\
\label[equiv]{comp-identity-fin-4}
&≃ ∑_{(α:A = B)} (\coe{α} ∘ f = g ∘ \coe{α}).
\end{align}
\end{subequations}

\end{proof}

\section{Fundamental Notions of Graph Theory}\label{sec:notions-of-graph-theory}

While graphs are ubiquitous in science, exactly what is referred to by
the word \say{graph} depends upon the context. In some settings,
graphs are undirected while in others, graphs are directed. Sometimes
self-edges are allowed, sometimes not. Which notion is chosen in a
given context depends on the application, e.g.~power graphs in
computational biology, quivers in category theory, and networks in
network theory.

\subsection{The Type of Graphs}\label{sec:type-of-graphs}

In this work, a graph is a directed multigraph (self-edges are
allowed). The type of graphs is the set-level structure of the notion
of abstract graphs.

\begin{definition}\label{def:graph} A graph is an
object of type $\Graph$. The corresponding data is a \emph{set} of
points in a space (also known as \emph{nodes}) and a set for each pair
of points (also known as \emph{edges}).

\begin{equation*}
\Graph :≡ \sum_{(\Node:\UU)}\sum_{(\Edge:\Node → \Node →
\UU)}\, \isSet{\Node} × \prod_{(x,y : \Node)} \isSet{\Edge(x,y)}.
\end{equation*}
\end{definition}

The initial graph examples are the trivial ones: the \emph{empty}
graph and the \emph{unit} graph defined, respectively as
\(K_0 :≡ (\mathbb{0}, λ\,u\,v.\mathbb{0})\) and
\(K_1 :≡ (\mathbb{1}, λ\,u\,v.\mathbb{1})\).

Given a graph \(G\), for brevity, the set of nodes and the family of
edges are denoted by \(\Node_{G}\) and \(\Edge_{G}\), respectively. In
this way, the graph \(G\) is defined as
\((\Node_{G}, \Edge_{G}, (s_{G}, t_{G})) : \mathsf{Graph}\) where
\(s_{G} : \isSet{\Node_{G}}\) and
\(t_{G} : \prod_{x,y:\Node_{G}} \isSet{\Edge_{G}(x,y)}\). We may refer
to \(G\) only as the pair \((\Node_{G}, \Edge_{G})\), unless we
require to show the remaining data, the mere propositions \(s_{G}\)
and \(t_{G}\). Additionally, we will use the variables \(G\) and \(H\)
to be graphs, and variables \(x,y\) and \(z\) to be nodes in \(G\),
unless stated otherwise.

\begin{definition}\label{def:graph-homomorphism} A \emph{graph
homomorphism} from $G$ to $H$ is a pair of functions $(α, β)$ such
that $α : \Node_{G} → \Node_{H}$ and $β : \prod_{x,y : \Node_{G}}
\Edge_{G}(x,y) → \Edge_{H}(α(x),α(y))$. We denote by $\Hom{G}{H}$ the
type of these pairs.
\end{definition}

We denote by \(\mathsf{id}_{G}\), for any graph \(G\), the identity
graph homomorphism where the corresponding \(\alpha\) and \(\beta\)
are the identity functions.

\begin{lemma}\label{lem:hom-is-set}
The type $\Hom{G}{H}$ forms a set.
\end{lemma}

\hypertarget{the-category-of-graphs}{%
\subsection{The Category of Graphs}\label{the-category-of-graphs}}

Graphs as objects and graph homomorphisms as the corresponding arrows
form a small pre-category. In fact, the type of graphs is a small
univalent category in the sense of the HoTT Book
\citep[§9.1.1]{hottbook}. This fact follows from
\Cref{thm:equivalence-principle} and, morally, because the \(\Graph\)
type is a set-level structure.

In a (pre-) category, an isomorphism is a morphism which has an
inverse. In the particular case of graphs, this can be formulated in
terms of the underlying maps being equivalences.

\begin{lemma}\label{def:isomorphism} Let $h$ be a graph homomorphism
given by the pair-function $(α, β)$. The claim $h$ is an
\emph{isomorphism}, denoted by $\mathsf{isIso}(h)$, is a proposition
equivalent to stating that the functions $\alpha$ and $β(x,y)$ for all
$x,y:\Node_{G}$, are both bijections.
\end{lemma}

The collection of all isomorphisms between \(G\) and \(H\) is denoted
by \(G \cong H\). If \(G \simeq H\), one says that \(G\) and \(H\) are
\emph{isomorphic}.

\begin{lemma}\label{lem:iso-is-set}
  The type $G \cong H$ forms a set.
\end{lemma}

We define a type to compare sameness in graphs in
\Cref{def:isomorphism}; the type of graph isomorphisms. In HoTT, the
identity type (\(=\)) serves the same purpose, and one expects
\citep{iso-implies-equality} the two notions coincide. In
\Cref{thm:equivalence-principle}, we prove that they are in fact
homotopy equivalent. The same correspondence for graphs also arises
for many other structures \citep{ahrens2018, Ahrens2020}, for example,
groups, and topological spaces.

\begin{theorem}[Equivalence principle]\label{thm:equivalence-principle}
The canonical map $(G = H) → (G \cong H)$ is an equivalence for $G , H : \Graph$.
\end{theorem}

\begin{proof}
  It is sufficient to show that $(G = H) \simeq (G\cong H)$. Remember
  that being an equivalence for a function constitutes a proposition.
  The equivalence follows from \Cref{proof:equivalence-principle}. We
  consider the type families $F_1(X):≡ X \to X \to \UU$ and
  $F_2(X,R):≡\prod_{x, y : X} \isSet{R(x,y)}$ to shorten the
  presentation.

\begin{subequations}
  \label[calc]{proof:equivalence-principle}
\begin{align}
(G = H) &≡ ((\Node_{G}, \Edge_{G}, (s_{G}, t_{G})) = (\Node_{H}, \Edge_{H}, (s_{H}, t_{H}))) \label{eq:equiv-principle-step-0}\\
&≃ \sum_{(α : \Node_{G} = \Node_{H})}
\sum_{(β : \tr{F_1}{α}{\Edge_{G}} = \Edge_{H})}
(\tr{\mathsf{isSet}}{α}{s_{G}} = s_{H}) × (\trtwo{F_2}{α}{β}{t_{G}} = t_{H})
\label{eq:equiv-principle-step-1}\\
&≃ \sum_{(α : \Node_{G} = \Node_{H})}
\sum_{(β : \tr{F_1}{α}{\Edge_{G}} = \Edge_{H})} \mathbb{1} \times \mathbb{1} \label{eq:equiv-principle-step-2}\\
&≃ ∑_{(α : \Node_{G} = \Node_{H})} (\tr{F_1}{α}{\Edge_{G}} = \Edge_{H})\label{eq:equiv-principle-step-3}\\
&≃ ∑_{(α : \Node_{G} = \Node_{H})} ∏_{(x,y:\Node_{G})}\Edge_{G}(x,y)=\Edge_{H}(\overline{α}(x),\overline{α}(y))\label{eq:equiv-principle-step-4}\\
&≃ ∑_{(α : \Node_{G} ≃ \Node_{H})} ∏_{(x,y:\Node_{G})}\Edge_{G}(x,y)≃\Edge_{H}(\overline{α}(x),\overline{α}(y))\label{eq:equiv-principle-step-5}\\
&≃ (G ≅ H).\label{eq:equiv-principle-step-6}
\end{align}
\end{subequations}
{\Crefname{equation}{Equivalence}{Equivalences}
We first unfold definitions in \Cref{eq:equiv-principle-step-0}.
\Cref{eq:equiv-principle-step-1} follows from the characterisation of
the identity type between pairs in a $\Sigma$-type (Lemma 3.7 in HoTT
book). \Cref{eq:equiv-principle-step-2} stems from the fact that being
a set is a mere proposition and, thus, equations between proofs of
such are contractible, similarly as in
\Cref{lem:cyclic-identity-type}. To get
\Cref{eq:equiv-principle-step-4}, we apply function
extensionality twice in the inner equality in
\Cref{eq:equiv-principle-step-3}. By the Univalence axiom, we replace
in \Cref{eq:equiv-principle-step-5} equalities by equivalences.
Finally, \Cref{eq:equiv-principle-step-6} follows from
\Cref{def:isomorphism} completing the calculation from which the
conclusion follows.
}
\end{proof}

\begin{lemma}\label{lem:graph-is-gropoid}
The type of graphs is a $1$-groupoid.
\end{lemma}

\begin{proof}
We want to show that the identity type $G = H$ is a set for all $G,H :
\Graph{}$. This follows since type equivalences preserve homotopy
levels. The type $G=H$ is equivalent to the set of isomorphisms, $G
\simeq H$, by the equivalence principle, \Cref{thm:equivalence-principle}.
\end{proof}

\subsection{Graph Classes}\label{sec:graph-classes}

Graphs can be collected in different classes. A \emph{class} \(C\) of
graphs is a collection of graphs that holds some given structure
\(P : \Graph\, \to \UU\), i.e.~a graph class is the total type of the
corresponding predicate, \(C :≡ Σ_{G : \Graph}\,P(G)\). Examples of
classes are \emph{simple} graphs where the edge relation is
propositional, or \emph{undirected} graphs where the edge relation is
symmetric. \Cref{def:graph} is the class of directed multigraphs,
which leads us to more general statements about graphs and to write
shorter proofs. It is more general, because one can see, for example
that simple undirected graphs are instances of directed multigraphs.

Now, since any construction in HoTT respects the structure of its
constituents, a graph class is invariant under graph isomorphisms.
Specifically, given a graph isomorphism, we can transport any property
on graphs along the equality obtained by
\Cref{thm:equivalence-principle}. Graph properties provide a way to
determine if negative statements attribute on certain graphs, for
example, if two given graphs are not isomorphic.

\begin{lemma}[Leibniz principle]\label{leibniz-principle}
  Isomorphic graphs hold the same properties.
\end{lemma}

\noindent A related principle is equivalence induction
\citep[§3.15]{Escardo2019}.

\begin{lemma}[Equivalence induction]\label{equivalence-induction}
Given a graph $G$ and a family of properties $P$ of type
$\sum_{H:\Graph} (G \cong H) → \hProp$, if the property $P(G,
\mathsf{id}_{G})$ holds then the property also holds for any
isomorphic graph $H$ to $G$, i.e. $P(H, φ)$ holds for all $φ : G ≅ H$.
\end{lemma}

Lastly, of importance for this work are the class of connected finite
graphs stated in \Cref{def:finite-graph,sec:connected-graphs}. We will
assume any graph in the remaining of this paper, as connected and
finite, unless stated otherwise.

\subsubsection{Finite graphs}\label{sec:finite-graphs}

A graph is \emph{finite} if both the node set and each edge-set are
finite sets. The corresponding graph property of finite graphs is
stated in \Cref{def:finite-graph}. Similarly, as with finite types, a
finite graph has associated a cardinal to the number of nodes and of
the number of edges. Consequently, one can show that for finite
graphs, the equality is decidable on the node set and on each
edge-set.

\begin{definition} \label{def:finite-graph}
A graph $G$ is \emph{finite} if $\mathsf{FiniteGraph}(G)$ holds.
\begin{equation*}
\mathsf{FiniteGraph}(G) :≡ \mathsf{isFinite}(\Node_{G}) \,
\times \prod_{(x,y : \Node_{G})} \mathsf{isFinite}(\Edge_{G}(x,y)).
\end{equation*}
\end{definition}

\subsubsection{Connected Graphs}\label{sec:connected-graphs}

Among the many notions in graph theory, the concept of a walk plays a
leading role for many computer algorithms. A \emph{walk} in a graph is
simply a sequence of edges that forms a chain of types stated in
\Cref{def:walk}. By considering the walks in a graph, one obtains an
endofunctor \(W\) in \(\Graph\) by mapping a graph \(G\) to a graph
formed by the same node set in \(G\), and the corresponding sets of
walks in \(G\) as the edge sets.

\begin{definition}\label{def:walk}
A \emph{walk} between $x$ and $y$, is an element of the type $\mathsf{Walk}(x,y)$,
which is inductively generated by
 \begin{itemize}
  \item the trivial walk $⟨x⟩$ when $x ≡ y$, and
  \item composite walks, $(e\,⊙\,w)$, where $e$ is an edge from $x$ to
  some $z$, and $w$ a walk from $z$ to $y$.
 \end{itemize}
\end{definition}

\begin{lemma} The type of walks between $x$ to $y$ forms a set.
\end{lemma}

\begin{definition}\label{def:connected-graph} A graph $G$ is
  (strongly) \emph{connected} when a walk merely exists in $G$ joining
  any pair of nodes. If $G$ is connected then the proposition
  $\Connected(G)$ holds.
  
  \begin{equation*} \Connected(G) :\equiv
    \prod_{(x,y : \Node_{G})} \| \Edge_{W(G)}(x,y)\|. 
  \end{equation*}
\end{definition}

\begin{lemma}\label{lem:connected-is-proposition}
  Being connected for a graph is a proposition.
\end{lemma}

\subsection{Graph Families}
\label{sec:families-of-graphs}

The path graph with \(n\) nodes is the graph
\(P_n :≡ ([n], λ\,u\,v. \mathbb{N}\mbox{-}\mathsf{succ}(\mathsf{to}\mathbb{N}(u)) = \mathsf{to}\mathbb{N}(v))\).
One alternative type of walks is the type of path graphs, where \(n\)
represents the length of the walk. In other words, a \emph{path} in
\(G\) of length \(n\) between nodes \(a\) and \(b\) is a graph
homomorphism from \(P_n\) to \(G\) mapping \(0\) to \(a\) and \(n-1\)
to \(b\). If \(a\) and \(b\) are equal, such a path is \emph{closed}.
A closed path induces a \emph{cycle} graph of size of the same length
of the path. The family of cycle graphs follows
\Cref{def:cycle-graph}.

\begin{definition}\label{def:cycle-graph} An $n$-cycle graph denoted
by $C_n$ is defined by $C_n :≡ ([n], λ\,u\,v.u = \fpred{v})$ if $n≥1$.
Otherwise, $C_0$ is the one-point graph. 
\end{definition}

\begin{center}
\includegraphics[width=0.9\textwidth]{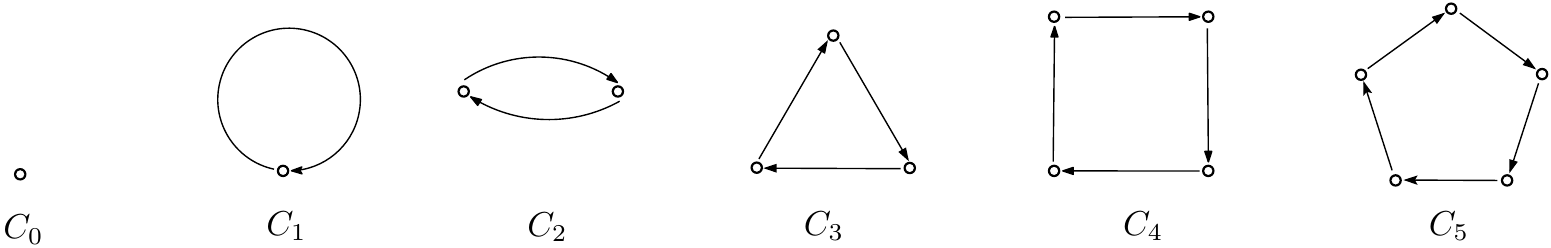}
\end{center}

In the treatment of embeddings of graphs on surfaces, we find out that
bouquet graphs, besides their simple structure, have nontrivial
embeddings, see \Cref{sec:examples-of-combinatorial-maps}.

\begin{definition} The family of \emph{bouquet} graphs $B_n$ consists of
graphs obtained by considering a single point with $n$ self-loops.
\begin{center}
\includegraphics[width=0.9\textwidth]{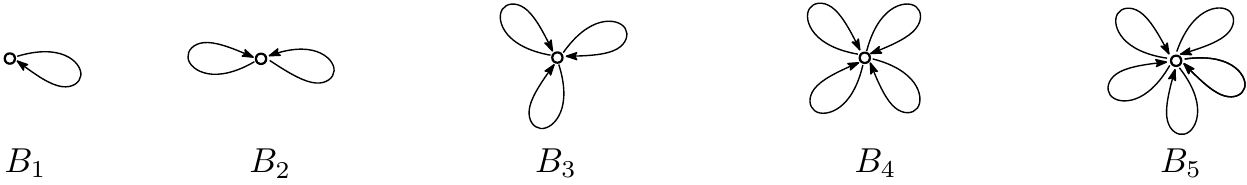}
\end{center}
\end{definition}

\begin{definition} A graph of $n$ nodes is called \emph{complete} when
every pair of distinct nodes is joined by an edge. In particular, we
denote by $K_{n}$, the complete graph with node set $[n]$. For brevity,
we use a double arrow in the pictures below to denote a pair of edges
of opposite directions.

\begin{center}
\includegraphics[width=0.9\textwidth]{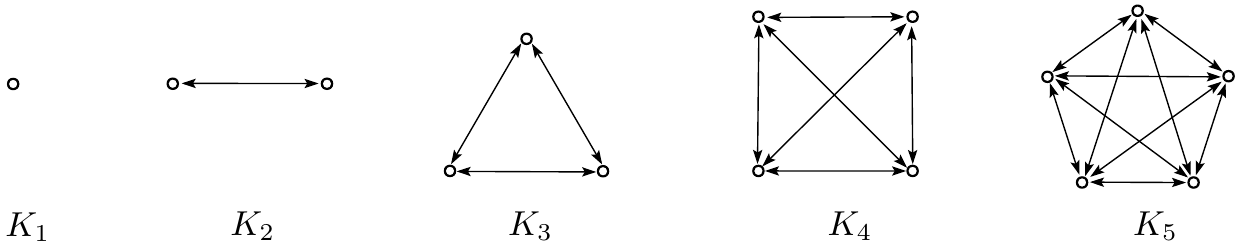}
\end{center}
\end{definition}

\subsection{Other Structures on Graphs}

\subsubsection{Cyclic Graphs}
\begin{wrapfigure}[6]{R}[0pt]{-5pt}
  \centering
\begin{tikzcd}
n-1 \arrow[d, "\mathsf{rot}^{n-i-1}"', dotted, maps to] & 0 \arrow[l, "\mathsf{rot}"', maps to] & 1 \arrow[l, "\mathsf{rot}"', maps to]              \\
i \arrow[rr, "\mathsf{rot}"', maps to]                  &                                       & i-1 \arrow[u, "\mathsf{rot}^{i-2}"', dotted, maps to]
\end{tikzcd}
\end{wrapfigure}

Similarly, as for cyclic types, we introduce a class of graphs with a
cyclic structure. A graph is \emph{cyclic} when it is in the connected
component of an \(n\)-cycle graph along with the corresponding
automorphism, see \Cref{def:connected-component}.

Now, let us consider the homomorphism
\(\mathsf{rot} : \Hom{C_n}{C_n}\) that acts similarly as the function
\(\mathsf{pred}\) in \Cref{def:cyclic-type}. The cyclic structure for
graphs can be defined as the property of preserving the structure in
\(C_n\) induced by the morphism \(\mathsf{rot}\). We will make use of
the same notation as for cyclic sets to refer to cyclic graphs.

\begin{definition}\label{def:cyclic-graph}
A graph $G$ is \emph{cyclic} if it is of type $\mathsf{CyclicGraph}(G)$.
\begin{equation*}
\mathsf{CyclicGraph} (G):≡ ∑_{(φ : \Hom{G}{G})} \sum_{(n : \N)}
\underbrace{\| (G , φ) = (C_{n}, \mathsf{rot})\|}_{\mathsf{iscyclic(G, \varphi,n)}}.
\end{equation*}
\end{definition}

\begin{figure} \centering
  \includegraphics[width=0.14\textwidth]{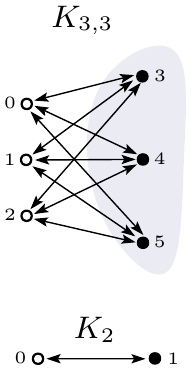}
  \caption{The graph $K_{3,3}$. Each arrow in the picture represents a
  pair of edges, one in each direction.}
\end{figure}

\subsubsection{Graph Colourings}\label{sec:graph-colourings}

A \(n\)-colouring of a graph \(G\) is a homomorphism of type
\(\Hom{G}{K_{n}}\), where each node in \(K_{n}\) represent a different
colour for the nodes in \(G\). Given an \(n\)-colouring of \(G\), we
say that \(G\) is \(n\)-colourable or \(n\)-partite. Thus, a
\emph{bipartite} graph is a graph with a \(2\)-colouring, and a
bipartite complete graph with six nodes is the graph \(K_{3,3}\). The
collection of all \(n\)-colourings of a graph forms a set by
\Cref{lem:hom-is-set}, and the collection of \(n\)-partite graphs
forms a \(1\)-groupoid. Since there are some \(n\)-partite graphs that
are equal up to isomorphism, we have the following distinction. Two
graph colourings of \(G\), namely \(f, g : \Hom{G}{K_n}\) are
\emph{essentially} different if a nontrivial isomorphism
\(\sigma : K_n \cong K_n\) exists and if the functions \(f\) and
\(g \circ \sigma\) are equal. The type of essentially different
colourings of a graph \(G\) is \Cref{def:n-essentially-partite-graph}.

\begin{equation}\label[type]{def:n-essentially-partite-graph}
\mathsf{EssentiallyPartite}(n, G) :≡ \sum_{(A : \Graph)} \Hom{G}{A} \times \trunc{A \cong K_{n}}.
\end{equation}

\begin{example}
We compute the identity type of the essential different colourings of
the path graph $P_3$ in \Cref{calc:essentiall-dif}. As we will see,
there can only be two graph homomorphisms from ${P_3}$ to $K_2$,
namely $\varphi_0$ and $\varphi_1$ as in \Cref{fig:example-n-partite}.
Let $c_1$ and $c_2$ be of type $\mathsf{EssentiallyPartite}(2,P_3)$.

\begin{subequations}\label[calc]{calc:essentiall-dif}
\begin{align}
(c_1 = c_2) &\simeq ((K_2,  \varphi_0, ! ) = (K_2,  \varphi_1, ! ))\\
&\simeq \sum_{(\tau: K_2 = K_2)} (\tr{\,\lambda X.\Hom{P_3}{X}}{\tau}{\varphi_0} = \varphi_1)\label[equiv]{calc:essentiall-dif-2}
\\
&\simeq \sum_{(\tau: K_2 = K_2)} (\tau \circ \varphi_0 = \varphi_1).
\end{align}
\end{subequations}
In \Cref{calc:essentiall-dif-2}, the equality $\tau : K_2 = K_2$ is
one of two alternatives: the trivial path or the path from the
equivalence that swaps the only two nodes in $K_2$. Only the latter
possibility, the equation $\tau \circ \varphi_0 = \varphi_1$ can hold.
\end{example}

\begin{figure}[!ht]
  \begin{center}
  \includegraphics[width=0.3\textwidth]{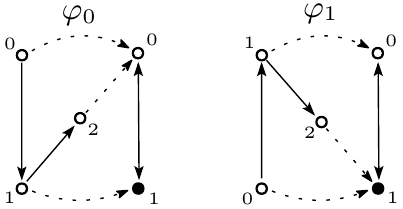}
  \caption{Two graph homomorphisms $\varphi_0$ and $\varphi_1$ from
  $P_{3}$ to $K_{2}$. The dashed arrows represent how $\varphi_0$ and
  $\varphi_1$ map the nodes of $P_3$ into $K_2$. We represent the colours of
  the $2$-coloring of $P_3$ by the nodes black and white in $K_2$.}
  \label{fig:example-n-partite}
  \end{center}
\end{figure}

\subsection{The Identity Type on Graphs}\label{sec:automorphisms}

For any element, \(x\) of a groupoid type, \(X\), the type
\(\mathsf{Aut}_X(x):= (x=x)\) has a group structure given by path
composition. Applying this definition to the groupoid of graphs, the
equivalence principle of \Cref{thm:equivalence-principle} gives that
for any graph \(G\), we identify \(\mathsf{Aut}(G)\) with its
automorphisms, \(G\cong G\). This allows us to compute
\(\mathsf{Aut}(G) = G\cong G\) in the examples which follow.

\begin{enumerate}
  \item $\Aut(B_2)$ is the group of two elements. With only two edges in $B_2$
  and one node, we can only have, besides the identity function, the function
  that swaps the two edges. In general, the identity type $B_n = B_n$ is
  equivalent to the group $S_n$, the group which contains the permutations of $n$
  elements.

  \item $\Aut(K_{3,3})$ is the subgroup $\mathbb{Z}_2\times S_3 \times
  S_3$ in $S_6$, since the nodes of $K_{3,3}$ can be partitioned into
  two sets of three, which can be permuted independently.
  Additionally, the two partitions are interchangeable.


  \item Any isomorphism in $\Aut(C_{n})$ is completely determined by
  how it acts on a fixed node in $C_n$. 
\end{enumerate}

\begin{lemma}\label{lem:automorphisms-of-Cn-equiv} Given the function
$\mathsf{rot}$ as in \Cref{def:cyclic-graph} and $k<n$, the function
$\mathsf{rot}^{k}$ is a bijection between $[n]$ and $(C_{n} \simeq
C_{n})$. Moreover, if $k_i < n$ for $i=1,2$ and
$\mathsf{rot}^{k_1} = \mathsf{rot}^{k_2}$ then $k_1 = k_2$.
\end{lemma}

\section{Graph Embeddings}\label{sec:graph-embeddings}

Graphs are commonly represented by their drawings on a surface like
the plane. In topology, such a drawing --- also called \emph{graph
embedding} --- can be represented as an embedding of the topological
realization of the abstract graph on some given surface
\citep{Stahl1978}. Not all finite graphs can be drawn in the plane,
but all finite graphs can be drawn on some orientable surface. If
\(G\) denotes the abstract graph, then we denote by \(|G|\) the
topological realization of \(G\).

Given a graph embedding in the surface \(S\), say \(f: |G| ↪ S\), one
can consider the space \(S - f(G)\), the surface with the image of the
graph removed. This space consists of a collection of connected
components. Such a component is called a \emph{face} if it is
homeomorphic to an open disk. If all connected components are faces,
one says that the graph embedding is \emph{cellular}
\citep[§3.1.4]{gross}.

Cellular embeddings are interesting because they can be characterised
combinatorially --- up to isotopy --- by the cyclic order which they
induce on the set of nodes around each node.

\subsection{Locally Finite Graphs}\label{sec:locally-finite}

A graph \(G\) is \emph{locally finite} if the set of \emph{incident}
edges at every node \(x\), also called the star of \(x\) in \(G\), is
a finite set. The \emph{valency} of a star in a locally finite graph
is the cardinal number of the corresponding set of edges. Now if we
consider the graph \(U(G)\) as the symmetrisation of \(G\), with
\(\Node_G\) as the node set, and edges between \(x\) and \(y\), of the
type \(\Edge_{G}(x,y) + \Edge_{G}(y,x)\), then the star at \(x\) is
\Cref{def:star}.

\begin{equation}\label[type]{def:star}
  \mathsf{Star}_{G}(x) :≡ \sum_{(y:\Node_{G})} \mathsf{Edge}_{U(G)}(x,y).
\end{equation}

\begin{lemma}\label{lem:star-is-set}
The stars at each node forms a set.
\end{lemma}
\begin{proof}
It follows since the base type of \Cref{def:star} is the set of edges
in the graph, and each of the fibers of the $\Sigma$-type is a
coproduct of sets, which we know they form sets \cite[Exer. 3.2]{hottbook}.
\end{proof}

It is proven that if the graph has at least one node of infinite
valency, such a graph can not have a cellular embedding into any
surface \citep[Proposition §3.2]{mohar1988}. In this work, we
therefore only consider locally finite graphs, and we also assume that
all graphs are connected, see \Cref{def:connected-graph}. One can
prove that if the graph \(G\) is connected, then \(U(G)\) is
connected.

\begin{figure}
  \centering
  \begin{subfigure}[b]{0.33\linewidth}
    \centering\includegraphics{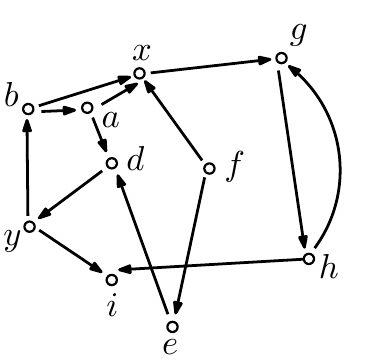}
    \caption{Graph $G$.}
  \end{subfigure}%
  \begin{subfigure}[b]{0.33\linewidth}
    \centering\includegraphics{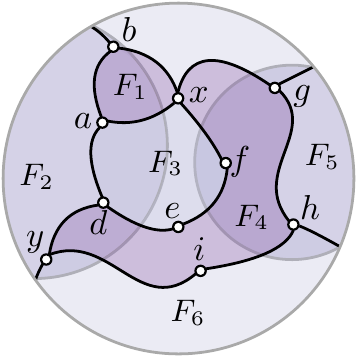}
    \caption{Embedded graph $U(G)$.}
  \end{subfigure}%
  \begin{subfigure}[b]{0.33\linewidth}
    \centering\includegraphics{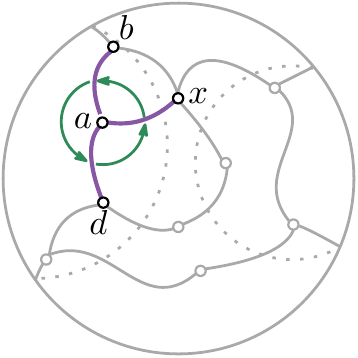}
    \caption{The rotation system at node $a$.}
  \end{subfigure}
  \caption{We show in (a) the drawing of a graph $G$ with edge
  crossings. A representation of the graph $G$ embedded in the sphere
  is shown in (b). The graph embedded $U(G)$ serves as the symmetrisation of 
  the graph $G$. Recall an edge $e$ in $G$ from $x$ to $y$ induces
  an edge in $U(G)$ from $x$ to $y$ and an edge from $y$ to $x$. For
  brevity, we only draw a segment representing such a pair of related
  edges. The corresponding faces of the graph embedding shaded in (b)
  are named $F_i$ for $i$ from $1$ to $6$. It is shown in (c) with
  {\color{color4!60!black} fuchsia} colour the incident edges at the
  node $a$ in $U(G)$. The rotation system at $a$, i.e. the cyclic set
  denoted by $(ba\,ax\,ad)$, is shown in {\color{darkgreen} green}
  colour. The dashed lines represent edges not visible to the view.}
   \label{fig:drawing-graph}
  \end{figure}

\subsection{The Type of Combinatorial Maps}\label{sec:combinatorial-maps}

Cellular embeddings are interesting because they can be characterised
combinatorially --- up to isotopy --- by the cyclic order which they
induce on the set of nodes around each node. We illustrate in
\Cref{fig:drawing-graph} (b) with a graph such a cyclic order. Since
all embeddings in the plane are cellular, we will only work with this
kind of embeddings. We forsake the topological definition of the
embedding to work with their combinatorial characterisation,
\Cref{def:combinatorial-map}.

\begin{definition}\label{def:combinatorial-map} A map for a graph $G$,
  of type $\mathsf{Map}(G)$, is a local \emph{rotation system} at each
  node.
  
  \begin{align*} \mathsf{Map}(G)  &:≡ \prod_{(x : \Node_{G})}
  \mathsf{Cyclic}(\mathsf{Star}_{G}(x)). \end{align*} \end{definition}

\begin{lemma}\label{lem:map-is-set}
The type of maps for a graph forms a set.
\end{lemma}

A graph with a map is a locally finite graph. Note that the stars of
any map are finite sets, because being cyclic for a type implies that
the type is a finite set.

\begin{lemma}
The type of maps for a finite graph is a finite set.
\end{lemma}

For brevity, we will use from now the variable \(\mathcal{M}\) to
denote a map of the graph \(G\).

\subsection{The Type of Faces}\label{sec:type-of-faces}

Combinatorially, a face consists of a cyclic walk in the embedded
graph where there are no edges on the inside of the cycle, and no node
occur twice. \Cref{def:face} is our attempt to make this intuition
formal. The criterion \say{no edges on the inside} for a face ---
called \emph{map-compatibility} below --- is captured by the fact that
each pair of sequential edges on the face is a successor--predecessor
pair in the cyclic order of the edges around their common node. Also,
note that our graphs are directed, and therefore, the underlying graph
of a face can consist of edges in any direction. Then, faces of a map
of a graph \(G\) are related to the graph \(U(G)\).

\begin{definition}\label{def:face} A \emph{face} of $\mathcal{M}$, of
type $\mathsf{Face}(G, \mathcal{M})$, consists of:
\begin{enumerate}
  \item a cyclic graph $A$, and
  \item a graph homomorphism $h$ given by $(\alpha, \beta)$ of type $\Hom{A}{U(G)}$, which is
  \begin{enumerate}
  \item \emph{edge-injective}, as in \Cref{def:edge-injective},
    \item \emph{star-compatible}, if the condition in \Cref{eq:star-condition} holds, and
    \begin{equation}\label[type]{eq:star-condition}
      \prod_{(x : \Node_{G})} \ \| \mathsf{Star}_{G}(\alpha x) \| \to \|\mathsf{Star}_{A}(x) \|,
    \end{equation}
    \item
    \emph{corner-compatible}, which is the evidence that
  $h$ is compatible with the edge-ordering given by the map $\mathcal{M}$ at the node
  $\alpha(x)$ and the edge-ordering coming from the star at that node
  $x$ in $A$.  The corresponding type is \Cref{eq:map-compatible}.
  \begin{align}\label[type]{eq:map-compatible}
    \begin{split} 
    & \mathcal{M}(α(x))\,\langle\, (α (\mathsf{pred}(x)) , \mathsf{flip}(\beta(\mathsf{pred}(x),x,a))\,\rangle
     \\
     & =_{\mathsf{Star}_{G}(\alpha(x))} (\,\alpha(\mathsf{suc}(x))\,,\,\beta(x, \mathsf{suc}(x),a^{+})\,).
    \end{split}
  \end{align}
  The \emph{previous edge} at $x$ is the
  edge $a : \Edge_{\Node_{A}}(\mathsf{pred}(x), x)$, and the edge
  \emph{after} $a$ is the edge denoted by $a^{+}$ of type
  $\Edge_{\Node_{A}}(x , \mathsf{suc}(x))$, see \Cref{fig:face}.
  \end{enumerate}
\end{enumerate}

\noindent Given a face, we refer to $\mathsf{ismapcomp}(h)(x)$ to the
witness conditions of \Cref{eq:star-condition,eq:map-compatible}. \Cref{fig:face}
illustrates part of the required data to define a face for the map of
$G$ given in \Cref{fig:drawing-graph} (b). \end{definition}

\begin{figure} \centering
  \includegraphics[width=0.6\textwidth]{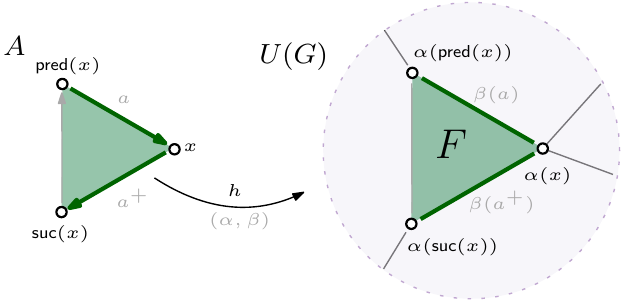} \caption{On the
  right side, we shade the face $F_1$ of the graph $G$ embedded in the
  sphere given in \Cref{fig:drawing-graph}. We have the cycle graph
  $C_3$ and $h:\mathsf{Hom}(C_3,U(G))$ given by $(\alpha, \beta)$ on
  the left side. $C_3$ and $h$ can be used to define the face $F_1$
  using $C_3$ as the graph $A$ in \Cref{def:face}.} \label{fig:face}
  \end{figure}

\begin{definition}\label{def:edge-injective}
  A graph homomorphism $h$ from $G$ to $H$ given by $(\alpha, \beta)$
  is \emph{edge-injective}, denoted by $\mathsf{isedgeinj}(h)$, if the
  function $f$ defined below is an embedding.
    \begin{align*}
        &f :   \sum_{(x,y : \Node_{G})} \Edge_{G}(x,y)
          \to \sum_{(x,y : \Node_{H})} \Edge_{H}(x,y)\\
        &f(x,y,e) :\equiv (\alpha(x), \alpha(y), \beta(x,y,e)).
    \end{align*}
\end{definition}

\begin{lemma}
  For a graph homomorphism,
  \begin{enumerate}
  \item being edge-injective is a proposition and
  \item being map-compatible is a proposition.
  \end{enumerate}
\end{lemma}

We devote the rest of this section to proving that the type of faces
forms a set in \Cref{lem:face-is-set}. This claim rests on the fact
that (i) the type of cyclic graphs forms a set, (ii) the type of graph
homomorphisms forms a set, and (iii) the conditions, edge-injective
and map-compatibe in \Cref{def:face} are mere proposition. One might
suspect this type forms a homotopy 1-groupoid from the previous facts.
However, the edge-injectivity property of the underlying graph
homomorphism of each face suffices to show that the type of faces is a
set.

\begin{lemma}\label{lem:edgeinj-implies-prop} Let $f$ and $g$ be
edge-injective graph homomorphisms from $C_{n}$ to a graph $G$ and
$n>0$. Then the type $\Sigma_{e : C_{n} = C_{n}}~(\tr{\lambda
X.\Hom{X}{G}}{e}{f} = g)$ is a mere proposition.
  \begin{equation*}
    \begin{tikzcd}[column sep=large]
        C_n \arrow[rd, "f"'] \arrow[rr, "e", equals] &   & C_n \arrow[ld, "g"] \\
                                                         & G &
        \end{tikzcd}
  \end{equation*}
\end{lemma}

\begin{proof}

The result follows from proving that the $\Sigma$-type in question is
equivalent to a proposition. The corresponding equivalence is given by
\Cref{lem:Cn-Cn-eq1}, in which we use some known results about
Univalence and \Cref{lem:automorphisms-of-Cn-equiv}, as in the very
last step.

  \begin{subequations}\label[calc]{lem:Cn-Cn-eq1}
    \begin{align}
    \sum_{(e : C_{n} = C_{n})}~(\tr{\lambda X.\Hom{X}{G}}{e}{f} = g) &\simeq \sum_{(e : C_{n} = C_{n})}~(f = g \circ \coe{e})\\
    &\simeq \sum_{(e : C_{n} \simeq C_{n})}~(f = g \circ e) \\
    &\simeq \sum_{(k : [n])} (f = g \circ \mathsf{rot}^{k}).
  \end{align}
\end{subequations}

It remains to show that the last equivalent type is a proposition. Let
$(k_1 , p_1)$, and $(k_2, p_2)$ be of type $\Sigma_{k : [n]} (f = g
\circ \mathsf{rot}^{k})$. We must show that  $(k_1 , p_1)$ is equal to
$(k_2, p_2)$. Since $\Hom{C_n}{G}$ is a set, we only need to prove
that $k_1$ is equal to $k_2$. To show that,
\Cref{lem:automorphisms-of-Cn-equiv} is used in the proof. Let us
consider the equality $(p_1^{-1} \cdot p_2)$ of type $(g \circ
\mathsf{rot}^{k_1} = g \circ \mathsf{rot}^{k_2})$. Then, by computing
the identity type of graph isomorphisms, one can obtain
$(p_1^{-1} \cdot p_2)$ is equivalent to having two equalities, namely
$p :\alpha(g \circ \mathsf{rot}^{k_1}) = \alpha(g \circ
\mathsf{rot}^{k_2})$ and $q$ of type
\begin{equation*}
   \tr{\lambda e.\prod_{(x,y : \Node_{C_n})} \Edge_{C_n}(x,y) → \Edge_{G}(e(x),e(y)) }
    {p}{g \circ \mathsf{rot}^{k_1}} = \beta(g \circ \mathsf{rot}^{k_2}).
\end{equation*}
By the characterisation of $\Sigma$-types and with the previous
equalities, $p$ and $q$, one can get another equality $r$ of
\Cref{lem:Cn-Cn-eq2} for $x,y:\Node_{C_n}$ and $e:\Edge_{C_n}(x,y)$.

\begin{equation}\label[type]{lem:Cn-Cn-eq2}
\begin{array}{lclll}
  &(\alpha (g \circ \mathsf{rot}^{k_{i}})(x)
    , & \alpha(g \circ \mathsf{rot}^{k_{i}})(y)
    , & \beta(g \circ \mathsf{rot}^{k_{i}})(x,y,e))\\
  &= (\alpha(g)(\alpha(\mathsf{rot}^{k_{i}})(x))
    , & \alpha(g)(\alpha(\mathsf{rot}^{k_{i}})(y))
    , & \beta(g)(\beta(\mathsf{rot}^{k_{i}})(x,y,e))
  ).
\end{array}
\end{equation}
Now since the graph homomorphism $g$ is edge-injective, applying
\Cref{def:edge-injective} to the equality $r$, one gets an equality
$r'$ of \Cref{lem:Cn-Cn-eq3}. By applying
\Cref{lem:automorphisms-of-Cn-equiv} to $r'$, we conclude that
$k_1$ is equal to $k_2$ from where the required conclusion follows.

\begin{equation}\label[type]{lem:Cn-Cn-eq3}
  \begin{array}{lcll}
    &(\alpha(\mathsf{rot}^{k_{1}})(x)
    , & \alpha(\mathsf{rot}^{k_{1}})(y)
    , & (\beta(\mathsf{rot}^{k_{1}})(x,y,e)))\\
    &= (\alpha(\mathsf{rot}^{k_{2}})(x)
    , & \alpha(\mathsf{rot}^{k_{2}})(y)
    , & (\beta(\mathsf{rot}^{k_{2}})(x,y,e))).
\end{array}
\end{equation}
\end{proof}

\begin{lemma}\label{lem:face-is-set}
  The faces of a map forms a set.
\end{lemma}

\begin{linkproof}[]{https://jonaprieto.github.io/synthetic-graph-theory/lib.graph-embeddings.Map.Face.isSet.html}
Let $F_1$ and $F_2$ be two faces of a map $\mathcal{M}$. We will show
that the type $F_1 = F_2$ is a mere proposition in
\Cref{proof:face-is-set}, with the following conventions.
\begin{itemize}
\item  $\mathcal{A}$ is the cyclic graph related to the face $F_1$, $\mathcal{A} :\equiv (A, (\varphi_{A},
n,\mathsf{iscyclic}(A,\varphi_{A},n)))$, and
\item $\mathcal{B}$ is the cyclic graph related to the face $F_2$, $\mathcal{B} :\equiv (B, (\varphi_{B},
m,\mathsf{iscyclic}(B,\varphi_{B},m)))$.
\end{itemize}

We first unfold definitions of $F_1$ and $F_2$ in
\Cref{proof:face-is-set-step2}, and simplifying propositions in
\Cref{proof:face-is-set-step3}, namely $\mathsf{isedgeinj}$,
$\mathsf{ismapcomp}$, and $\mathsf{iscyclic}$. Then, by expanding the
definitions of $\mathcal{A}$ and $\mathcal{B}$ in
\Cref{proof:face-is-set-step4}, and simplifying the propositions terms
such as being a cyclic graph, one gets \Cref{proof:face-is-set-step5}.
Next, we reorder in \Cref{proof:face-is-set-step5} the tuple
equalities to create an opportunity for path induction towards the
application of \Cref{lem:edgeinj-implies-prop}. Now, since we want to
prove that the type of faces is a set, and that itself is a
proposition, the truncation elimination is applied to the propositions
$\mathsf{iscyclic}(A,\varphi_{A},n)$ and
$\mathsf{iscyclic}(A,\varphi_{A},n)$. Then, the graphs $A$ and $B$
become respectively $C_n$ and $C_m$ in \Cref{proof:face-is-set-step6}.
\Cref{proof:face-is-set-step7} follows from the characterisation of the
identity type between tuples in a nested $\Sigma$-type.

\begin{subequations}
  \label[calc]{proof:face-is-set}
\begin{align}
  & (F_1 = F_2) \nonumber \\
  &\equiv ((\mathcal{A}, f, \mathsf{isedgeinj}(f), \mathsf{ismapcomp}(f))
         =(\mathcal{B}, g, \mathsf{isedgeinj}(g), \mathsf{ismapcomp}(g))
         )  \label[eq]{proof:face-is-set-step2} \\
  &\simeq ((\mathcal{A}, f) = (\mathcal{B} , g))  \label[equiv]{proof:face-is-set-step3} \\
  &\equiv
  (
    ( (A, (\varphi_{A}, n ,\mathsf{iscyclic}(A,\varphi_{A},n)))
    , f)
  = ( (B, (\varphi_{B}, m ,\mathsf{iscyclic}(B,\varphi_{B},m)))
      , g) \label[eq]{proof:face-is-set-step4} \\
  &\simeq
  (
    ( (A, (\varphi_{A}, n))
    , f)
  = ( (B, (\varphi_{B}, m))
      , g)  \label[equiv]{proof:face-is-set-step5} \\
  &\simeq ((n , ((C_{n},f) , \varphi_{C_{n}})) = (m , ((C_{m}, g) , \varphi_{C_{m}}))) \label[equiv]{proof:face-is-set-step6} \\
  &\simeq
    \sum_{(p : n = m)}
      \sum_{(e',-) : \sum_{(e : C_{n} = C_{m})} \tr{\lambda X .\Hom{X}{U(G)}}{e}{f} = g}
        (\tr{\lambda X . \Hom{X}{X}}{e'}{\varphi_{C_{n}}} = \varphi_{C_{m}}).
  \label[equiv]{proof:face-is-set-step7}
\end{align}
\end{subequations}

It only remains to show that \Cref{proof:face-is-set-step7} is a mere
proposition. We show this by proving that each type in
\Cref{proof:face-is-set-step7} is a proposition. First, we unfold the
cyclic graph definition for $C_n$ and $C_m$, using
\Cref{def:cyclic-graph}. Secondly, a case analysis on $n$ and $m$ is
performed. This approach creates four cases, where $n$ and $m$ can be
zero or positive. However, we only keep the cases where $n$ and $m$ are
structurally equal. One can show the other cases are
\href{https://jonaprieto.github.io/synthetic-graph-theory/lib.graph-embeddings.Map.Face.isSet.html\#860}{imposible}
with an equality between $n$ and $m$.

\begin{enumerate}
\item If both, $n$ and $m$, are zero,
then, by definition, $C_n$ and $C_m$ are the one-point graph. In this
case, the conclusion easily follows. The base type $n = m$ of the
total space in \Cref{proof:face-is-set-step7} is a proposition because
$\mathbb{N}$ is a set. The type $C_0 = C_0$ is a proposition since it
is contractible. The identity graph homomorphism is the
\href{https://jonaprieto.github.io/synthetic-graph-theory/lib.graph-embeddings.Map.Face.isSet.html\#2790}{unique}
automorphism of $C_0$. Lastly, because $\Hom{C_n}{C_n}$ is a
set, the remaining type of the $\Sigma$-type is a mere proposition, completing the
proof obligations.

\item If $n$ and $m$ are positive, we reason similarly. The type $n =
m$ is a proposition. By path induction on $p : n = m$, the second base
type of the $\Sigma$-type becomes \Cref{proof:face-is-set-step8}.

\begin{equation}\label[type]{proof:face-is-set-step8}
  \sum_{(e : C_{n} = C_{n})} (\tr{\lambda X .\Hom{X}{U(G)}}{e}{f} = g).
\end{equation}

\Cref{proof:face-is-set-step8} is a proposition by
\Cref{lem:edgeinj-implies-prop}. The remaining type of the
$\Sigma$-type is a mere proposition, because $\Hom{C_n}{C_n}$ is a
set. Therefore, the $\Sigma$-type in \Cref{proof:face-is-set-step7} is a
proposition as required.
 \end{enumerate}

\end{linkproof}

\subsection{Face Boundary}\label{sec:boundary-of-a-face}

Each face \(\mathcal{F}\) of a map \(\mathcal{M}\) given by
\(\langle A , h\rangle\) is bounded by a closed walk in \(U(G)\)
induced by the non-empty cyclic graph \(A\) through \(h\). We refer to
such a walk as the \emph{boundary} of \(\mathcal{F}\), and it is
denoted by \(\partial \mathcal{F}\). The \emph{degree} of
\(\mathcal{F}\) is the length of \(\partial \mathcal{F}\), which is
the number of nodes in \(A\). The boundary \(\partial \mathcal{F}\)
can be walked in two directions with respect to the orientation given
by its map.

As illustrated by \Cref{fig:walks-in-a-face}, given two different
nodes \(x\) and \(y\) in \(\partial \mathcal{F}\), we can connect
\(x\) to \(y\) using the walk in the clockwise direction,
\(\mathsf{cw}_{\mathcal{F}}(x,y)\). Similarly, one can connect \(x\)
to \(y\) using the walk in the counter-clockwise direction,
\(\mathsf{ccw}_{\mathcal{F}}(x,y)\). Such walks are induced by the
walks in the cyclic graph \(A\), see \Cref{lem:walk-in-UCn}.

\begin{lemma} \label{lem:walk-in-Cn}
  Supposing $x,y:\Node_{C_{n}}$, the following claims hold for the cycle graph
  $C_{n}$.
  \begin{enumerate}
    \item The type $\Edge_{C_{n}}(x,y)$ is a proposition.
    \item There exists an edge of type
     $\Edge_{C_{n}}(\mathsf{pred}(x), x)$ and an edge of type
     $\Edge_{C_{n}}(x, \mathsf{succ}(x))$.
    \item There exists a walk going in the clockwise
      direction denoted by $\mathsf{cw}_{C_n}(x,y)$ from $x$ to $y$. 
  \end{enumerate}
\end{lemma}

\begin{lemma} \label{lem:walk-in-UCn}
  Supposing $x,y:\Node_{C_{n}}$, the following claims hold for the cyclic graph
  $U(C_n)$.
  \begin{enumerate}
  \item If $n > 2$, then the type $\Edge_{U(C_{n})}(x,y)$ is a proposition.
  \item There exists an edge of type $\Edge_{C_{n}}(\mathsf{pred}(x), x)$ and of type
     $\Edge_{C_{n}}(x, \mathsf{succ}(x))$.
    \item There exist two quasi-simple\footnote{In a quasi-simple walk
    no repetitions of nodes occur, except its head, see Definition
    4.7 in \cite{homotopywalks}.} walks from $x$ to $y$:
    \begin{enumerate}
          \item There is a walk in the clockwise direction denoted by $\mathsf{cw}_{U(C_{n})}(x,y)$.
          \item If $x \neq y$, then the walk denoted by
           $\mathsf{ccw}_{U(C_{n})}(x,y)$ is the walk in the
           counter-clockwise direction from $x$ to $y$. Otherwise, if
           $x = y$, the walk $\mathsf{ccw}_{U(C_{n})}(x,y)$ is the
           trivial walk $\langle x \rangle$.
    \end{enumerate}
  \end{enumerate}
  \end{lemma}

\begin{lemma}\label{lem:cover-UCn}
   Supposing $x,y:\Node_{C_{n}}$, the following claims hold for the cyclic graph
  $U(C_n)$.
    \begin{enumerate}
      \item If $x \neq y$, then $n$ is the length of each walk in Item (3) in \Cref{lem:walk-in-UCn}.
      \item Otherwise, the walk
      $\mathsf{cw}_{U(C_{n})}(x,x)$ is of a length $n$,
      and  the walk $\mathsf{ccw}_{U(C_{n})}(x,x)$ is of length $0$.
    \end{enumerate}
\end{lemma}

\begin{figure}
\centering
\includegraphics[width=0.3\textwidth]{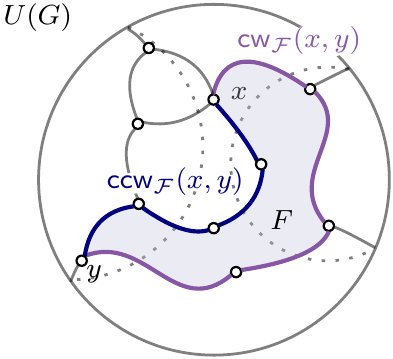}
\caption{It is shown a face $\mathcal{F}$ given by $\langle A ,
f\rangle$ for the graph embedding $U(G)$ given in
\Cref{fig:drawing-graph}. There are two quasi-simple walks in the
underlying cyclic graph $A$ between two different nodes $x$ and $y$.
Such walks are the clockwise and counter-clockwise
closed walks in $U(G)$, denoted by $\mathsf{cw}_{\mathcal{F}}(x,y)$ and
$\mathsf{ccw}_{\mathcal{F}}(x,y)$, respectively.}
\label{fig:walks-in-a-face}
\end{figure}

Additionally, one can prove that \(n\) is the maximum length of a
quasi-simple walk in the cyclic graph \(U(C_n)\), as in Lemma 4.13 in
\citep{homotopywalks}. Lastly, as illustrated in
\Cref{fig:walks-in-a-face} for the face \(\mathcal{F}\), the graph
\(U(C_n)\) is completely covered by the walks
\(\mathsf{ccw}_{U(C_{n})}(x,y)\) and \(\mathsf{cw}_{U(C_{n})}(x,y)\).

\subsection{Examples of Graph Maps}\label{sec:examples-of-combinatorial-maps}

In this subsection, we examine the combinatorial maps for the bouquet
graph \(B_2\) and for the family of cycle graphs \(C_n\) from
\Cref{sec:families-of-graphs}. We also give a map for the graph
\(K_{3,3}\) mentioned in \Cref{sec:graph-colourings}.

Let us introduce some notation for readability. An edge
\(a : \Edge_{G}(x,y)\) induces the edge \(\eout{a}\) in \(U(G)\) from
\(x\) to \(y\). Similarly, the edge \(a : \Edge_{G}(y,x)\) induces the
edge \(\ein{a}\) in \(U(G)\) from \(x\) to \(y\). Lastly, a face of a
map is given by the corresponding cyclic graph in \Cref{def:face}.

\begin{figure}
  \centering
    \includegraphics[width=0.7\textwidth]{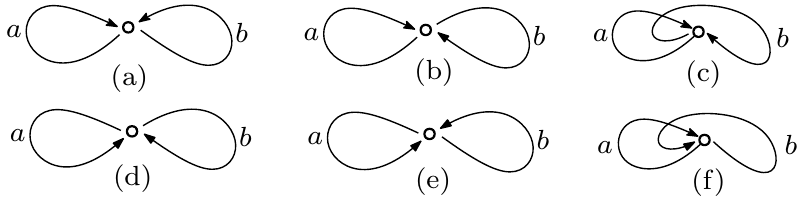}
   \caption{The six possible maps of the bouquet $B_2$.}
  \label{fig:map-bouquet}
\end{figure}
\begin{example}\label{ex:bouquet-maps} In \Cref{fig:map-bouquet}, we
  show the six combinatorial maps for the bouquet graph $B_2$. It is
  important to note that combinatorial maps are equal up to
  isomorphism. The equality of rotation systems at each node is
  given by \Cref{lem:cyclic-identity-type}. Thus, the bouquet $B_2$
  has only three distinct maps called $M_a$, $M_b$, and $M_c$. The
  pairs of equal maps are respectively $(a,d)$, $(b,e)$, and $(c,f)$
  in \Cref{fig:map-bouquet}.
\begin{align*}
  &M_{a}:\equiv (0 \mapsto (\eout{a}\ein{a}\ein{b}\eout{b})).\\
  &M_{b}:\equiv (0 \mapsto (\eout{a}\ein{a}\eout{b}\ein{b})).\\
  &M_{c}:\equiv (0 \mapsto (\eout{a}\eout{b}\ein{a}\ein{b})).
\end{align*}
The surface arising from the maps $M_a$ and $M_b$ is the
two-dimensional plane. For the map $M_c$, the surface is the
topological torus. We recall that the torus is a surface homeomorphic
to the cartesian product of a circle with itself, as the
embeddings $c$ and $f$ in \Cref{fig:map-bouquet}.
\end{example}

\noindent 

\begin{example}\label{ex:map-of-Cn} For cycle graphs $C_n$, only one
combinatorial map exists. To show this, we consider the equivalence
given by the function $f_{x}$ for all $x : C_n$ in
\Cref{func:example-map-4}. Then, any star in $C_n$ consists of two
different edges.

\begin{equation}\label{func:example-map-4}
\begin{split}
    &f_{x} : \mathsf{Star}_{C_n}(x) \to [2]\\
    &f_{x} (y, \mathsf{inl}(p)) :\equiv 0,\\
    &f_{x} (y, \mathsf{inr}(p)) :\equiv 1.
\end{split}
\end{equation}

If we consider the cyclic structures of the two-point type, $c_1
:\equiv \langle [2], \mathsf{pred}, 2\rangle$ and $c_2 :\equiv \langle
[2], \mathsf{succ}, 2\rangle$, then they induce precisely the maps of
$C_n$. In other words, one can obtain a map $\mathcal{M}$ using $c_1$
by \Cref{eq:Star-is-2-in-Cn} and $(\mathsf{pred}, 2, |(\mathsf{ideqv},
\mathsf{refl}_{\mathsf{pred}}) | ) : \mathsf{Cyclic}([2])$. By
\Cref{lem:cyclic-identity-type} implies that the map induced by $c_2$
and the map $\mathcal{M}$ are equal, by function extensionality.

\begin{equation}\label[equiv]{eq:Star-is-2-in-Cn}
  \begin{split}
  \mathsf{Map}(C_n) &\equiv \prod_{(x : [n])} \mathsf{Cyclic}(\mathsf{Star}_{C_{n}}(x))\\
  &\simeq \prod_{(x :[n])} \mathsf{Cyclic}([2]).
  \end{split}
\end{equation}
\end{example}

\begin{figure}
  \centering
  \includegraphics[width=.24\textwidth]{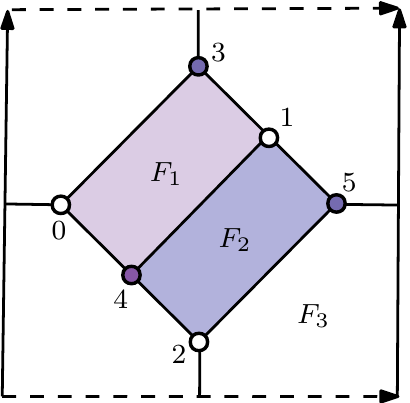}
  \caption{A map for $K_{3,3}$ in the surface of the torus.}
  \label{fig:map-k33}
\end{figure}

\begin{example}\label{ex:k33-map}

We below define for $K_{3,3}$ a map $\mathcal{M}$ and its
faces $F_1, F_2$, and $F_3$. The corresponding surface to
$\mathcal{M}$ is that of the torus, as illustrated by its polygonal
schema in \Cref{fig:map-k33}.

\begin{fleqn}[\parindent]
\begin{equation*}
\begin{split}
  \mathcal{M} :\equiv (
    &0 \mapsto ((03)\,(04)\,(05)) , 1 \mapsto ((13)\,(15)\,(14)), \\
    &2 \mapsto ((24)\,(25)\,(23)) , 3 \mapsto ((32)\,(31)\,(30)), \\
    &4 \mapsto ((40)\,(41)\,(42)), 5 \mapsto ((51)\,(50)\,(52))).\\
  F_{1} :\equiv (& (30)\, (04)\, (41)\, (13)).\\
  F_{2} :\equiv (& (14)\, (42)\, (25)\, (51)).\\
  F_{3} :\equiv (& (03)\, (32)\, (24)\, (40)\,(05)\, (52)\, (23)\, (31)\,(15)\,(50)).\\
\end{split}
\end{equation*}
\end{fleqn}
\end{example}

\section{Planar Embeddings} \label{sec:planar-embeddings}

In this section, we examine the class of graphs with an embedding in
the two-dimensional plane. Such embeddings are called \emph{planar
embeddings} or \emph{planar maps}. A graph is \emph{planar} if it has
a planar embedding and the graph embedded is called a \emph{plane}
graph. To discuss the notion of planar embeddings, we take inspiration
from topological graph theory \citep[§3]{gross}. Then one can work
with combinatorial maps that represent graph embeddings into a surface
---up to isotopy. In the following, we focus on describing embeddings
of graphs in the sphere called spherical maps. These maps are used
later to establish the type of planar embeddings for a given graph.

\subsection{Spherical Maps} \label{sec:spherical-maps}

Any graph embedding gives rise to an implicit surface. For planar
embeddings, this surface is a space homeomorphic to the sphere. In
particular, any embedding in the sphere induces an embedding in the
plane. To see this, for a graph embedded in the sphere, one can
puncture the sphere at some distinguished point, and subsequently,
apply the \emph{stereographic projection} to it.

The sphere in topology has two main invariants: path-connectedness and
simply-connectedness. The former states that a path connects any pair
of points in the sphere, and the latter states that any two paths with
the same endpoints in the sphere can be deformed into one another.

If we now consider a walk as a path in the corresponding space induced
by the map, then the path-connectedness property coincides with being
connected for the graph embedded. However, if we want to address
simply-connectedness for the surface induced by a graph-embedding,
then we need to have an equivalent notion to saying how a pair of
walks can be deformed into one the other. One proposal of such a
notion is \emph{homotopy for walks} in directed multigraphs
\citep{homotopywalks}.

\subsubsection{Homotopy for Walks}\label{sec:homotopy-for-walks}

Given a map \(\mathcal{M}\) for a graph \(G\), we consider the
relation \((\sim_{M})\) on the set of walks defined in
\citep{homotopywalks}. This relation is a congruence relation on the
category of objects induced by the endofunctor \(W\). The relation
\((\sim_{M})\) states in which way one can transform one walk into
another considering adjacent faces. Then the relation strictly depends
on the map \(\mathcal{M}\).

\begin{definition}\label{def:congruence-relation} Let $w₁,w₂$ be two
  walks from $x$ to $y$ in $U(G)$. The expression
  $\HomWalk{\mathcal{M}}{w₁}{w₂}$ denotes that one can \emph{deform}
  $w_1$ into $w_2$ along the faces of $\mathcal{M}$. We acknowledge
  the evidence of this deformation as a walk-homotopy between $w_1$
  and $w_2$, of type $\HomWalk{\mathcal{M}}{w₁}{w₂}$. The relation
  $(\sim_{\mathcal{M}})$ has four constructors as follows. The first
  three constructors are functions to indicate that homotopy for walks
  is an equivalence relation, i.e. $\mathsf{hrefl}$,
  $\mathsf{hsym}$, and $\mathsf{htrans}$. The fourth constructor,
  illustrated in \Cref{fig:constructors-for-homotopic-walks}, is the
  $\mathsf{hcollapse}$ function that establishes the homotopy
  \begin{equation*}\label{eq:collapse} \HomWalk{\mathcal{M}}{(w₁ \cdot
  \mathsf{ccw}_{\mathcal{F}}(a,b) \cdot w₂)} {(w₁ \cdot
  \mathsf{cw}_{\mathcal{F}}(a,b) \cdot w₂)}, \end{equation*} supposing
  one has the following,
  \begin{itemize} \item[(i)] a face
  $\mathcal{F}$ given by $\langle A, f \rangle$ of the map
  $\mathcal{M}$, 
  \item[(ii)] a walk $w₁$ of type
  $\mathsf{W}_{U(G)}(x,f(a))$ for $x : \Node_{G}$ with one node
  $a:\Node_{A}$, and 
  \item[(iii)] a walk $w₂$ of type
  $\mathsf{W}_{U(G)}(f(b),y)$ for $b :\Node_{A}$ and $y : \Node_{G}$.
  \end{itemize}
  \end{definition}

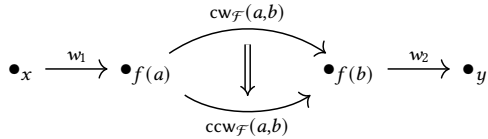
\begin{figure}[!ht]
  \centering
\[\begin{tikzcd}[column sep=normal]
{\bullet_{x}} & {\bullet_{f(a)}} && {\bullet_{f(b)}} & {\bullet_{y}}
\arrow["{w_1}", from=1-1, to=1-2]
\arrow[""{name=0, anchor=center, inner sep=0}, "{\mathsf{ccw}_\mathcal{F}(a,b)}"', curve={height=18pt}, from=1-2, to=1-4]
\arrow["{w_2}", from=1-4, to=1-5]
\arrow[""{name=1, anchor=center, inner sep=0}, "{\mathsf{cw}_\mathcal{F}(a,b)}", curve={height=-18pt}, from=1-2, to=1-4]
\arrow[shorten <=5pt, shorten >=5pt, Rightarrow, from=1, to=0]
\end{tikzcd}\] \caption{Given a face $\mathcal{F}$ of a map
$\mathcal{M}$, we illustrate here $\mathsf{hcollapse}$, one of the
four constructors of the homotopy relation on walks in
\Cref{def:congruence-relation}. The arrow $(\Downarrow)$ represents a
homotopy of walks.}
\label{fig:constructors-for-homotopic-walks}
\end{figure}

One consequence of \Cref{def:congruence-relation} is that, in each
face \(\mathcal{F}\), there is a walk-homotopy between
\(\mathsf{ccw}_{\mathcal{F}}(x,y)\) and
\(\mathsf{cw}_{\mathcal{F}}(x,y)\) using the constructor
\(\mathsf{hcollapse}\).

\subsubsection{The Type of Spherical Maps}\label{sec:type-of-spherical-maps}

As a property of maps, we can now state under which conditions the
surface arising from a map is simply-connected. We call such maps as
\emph{spherical maps}.

\begin{definition}\label{def:spherical-embedding} A map $\mathcal{M}$
  of a graph $G$ is \emph{spherical}, of type
  $\mathsf{Spherical}(\mathcal{M})$, if any pair of walks sharing the
  same endpoints are merely walk-homotopic.
  \begin{equation*}
  \mathsf{Spherical}(\mathcal{M}) :≡ \prod_{(x,y : \Node_{G})}
  \prod_{(w₁, w₂ : \mathsf{Walk}_{U(G)}(x,y))}
∥ \HomWalk{\mathcal{M}}{w₁}{w₂} ∥.
  \end{equation*}
  \end{definition}

\begin{lemma}\label{lem:spherical-is-proposition}
  Being spherical for a map is a proposition.
\end{lemma}

\begin{lemma}
The collection of all spherical maps for a (finite) graph is a
(finite) set.
\end{lemma}

\subsection{The Type of Planar Maps}\label{sec:type-of-planar-maps}

\begin{definition}\label{def:planar-graph} A planar map
$\mathcal{M}$ of a connected and locally finite graph $G$ is of type
$\mathsf{Planar}(G)$.
\begin{equation*}
\mathsf{Planar}(G) :≡ \sum_{(\mathcal{M} : \Map{G})} \mathsf{Spherical}(\mathcal{M}) ×
\underbrace{\Face{G,\mathcal{M}}}_{\text{outer face}}.
\end{equation*}
\end{definition}

\begin{theorem}\label{thm:being-planar-is-set}
  The type of all planar maps of a graph forms a set.
\end{theorem}

\begin{proof}
The type of planar embeddings in \Cref{def:planar-graph} is not a
proposition. It encompasses two sets: the set of combinatorial maps,
see \Cref{lem:map-is-set}, and the set of faces, see
\Cref{lem:face-is-set}. Since being spherical for a map is a mere
proposition, one concludes that the $\Sigma$-type collecting all
planar maps of a graph forms a set.
\end{proof}

\begin{example}\label{lem:Cn-is-planar} Let us prove that there exists
a planar map for $C_n$ with $n > 0$. Consequently, there exists a
planar map for every cyclic graph. Beside their simple structure,
cyclic graphs are building blocks in a few relevant constructions in
formal systems related to the study of planarity of graphs, as planar
triangulations using $C_3$, or a characterisation of all $2$-connected
planar graphs.

The graph $C_n$ is connected and locally finite, which mostly follows
from \Cref{lem:walk-in-Cn}. We must then show that $C_n$ has at least
one spherical embedding and one outer face. As described in
\Cref{ex:map-of-Cn}, there is only one such map that we denote here by
$\mathcal{M}$. This map gives rise to two faces, $F_1$ and $F_2$, the
inner face and the outer face, respectively. As the cycle graph $C_n$
is a finite graph, it is only required to consider the finite set of
quasi-simple walks to show that $\mathcal{M}$ is spherical, see Lemma
5.8 in \cite{homotopywalks}. The set of such walks is precisely given
in \Cref{lem:walk-in-UCn}. We must now show that any pair of such
walks are homotopic, from where one can conclude that the map
$\mathcal{M}$ is spherical, and consequently planar with outer face
$F_2$.

  \begin{enumerate}
  \item If $n=1$, the only walk to consider is the trivial walk, which is
    homotopic to itself.
  \item If $n>1$ and $x \neq y$, then we only need to consider the
  quasi-simple walks $\mathsf{ccw}_{U(C_{n})}(x,y)$ and
  $\mathsf{cw}_{U(C_{n})}(x,y)$. Such walks are homotopic by
  $\mathsf{hcollapse}(F_1, x,y,x,y,\langle x \rangle, \langle y
  \rangle)$.
  \item Otherwise, if $n>1$ and $x = y$, the only walks to consider
  are the trivial walk at $x$ and $\mathsf{cw}_{U(C_{n})}(x,x)$.
  Remember that the $\mathsf{ccw}_{U(C_{n})}(x,y)$ is by definition
  $\langle x \rangle$. Similarly, as in the previous case, these two
  walks are homotopic by the constructor $\mathsf{hcollapse}$.
  \end{enumerate}
\end{example}

\subsection{Planar Extensions}\label{sec:planar-extensions}

In this subsection, we describe how to construct a planar map from
another planar map. The characterisation of \(2\)-connected graphs
\citep{Whitney1932}, ear decompositions \citep[§5.3]{BangJensen2009},
reliable networks, and planar graph extensions for undirected graphs
\citep[§5.2,7.3]{Gross2018} are related constructions. In the current
section, \(G\) is a locally connected finite graph with decidable
equality on the set of nodes. For brevity, the variables \(p\) and
\(p_i\) will represent finite path graphs of a positive length.

\subsubsection{Path Additions}\label{sec:path-additions}

An \emph{internal node} of a path is any node that is not an endpoint
of the path. A \emph{simple path addition} to \(G\) is the graph
formed by adding \(p\) to \(G\) between two existing nodes of \(G\),
such that the edges and internal nodes of \(p\) are not in \(G\). A
\emph{simple cyclic
addition} is the addition of a cyclic graph to \(G\) with exactly one
node in common with \(G\). A \emph{non-simple} path addition is the
path addition of the graph \(U(q)\) to \(G\) for a path graph \(q\).
Similarly, one can define \emph{non-simple} cyclic additions. The
simple and non-simple addition of \(p\) to \(G\) are denoted by
\(G \bullet p\) and \(G \bullet \overline{p}\), respectively, and are
referred as \emph{graph extensions}. The operator \((\bullet)\) is
regarded as a left-associative operator.

Hereinafter, the path \(p\) in the addition \(G\bullet p\) proceeds
from \(u\) to \(v\) and its length is \(n+1\). The construction of
\(G\bullet p\) is equivalent to adding a path \(p'\) with two
distinguished edges \(\hat{u}\) and \(\hat{v}\), as illustrated in
\Cref{fig:planar-additions}b. If \(n=0\), the edges \(\hat{u}\) and
\(\hat{v}\) are equal. Or else, we have one edge from \(u: \Node_{G}\)
to \(0 : \Node_{p'}\), and another edge from \(n-1 : \Node_{p'}\) to
\(v\). The graph \(G\bullet p\) is formed by the set of nodes
\(\Node_{G} + \Node_{p'}\), and the corresponding edges, i.e.~the set
of edges in \(G\), \(\hat{u}\), \(\hat{v}\), and the set of edges in
\(p'\).

\begin{figure}[!ht]
  \centering
    \begin{subfigure}[b]{0.36\linewidth}
      \centering\includegraphics{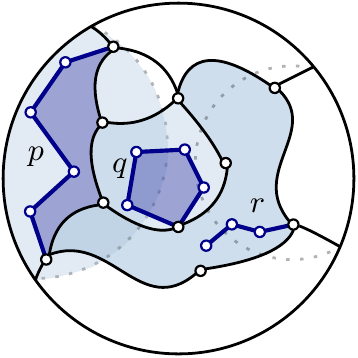}
      \caption{The embedded graph $U(G \bullet p \bullet q \bullet r)$.}
    \end{subfigure}%
    \begin{subfigure}[b]{0.36\linewidth}
      \centering
      \includegraphics{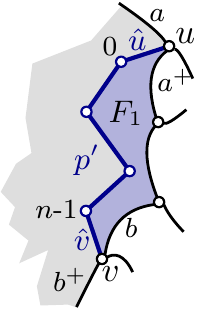}
      \caption{Addition of $p$ to $G$.}
    \end{subfigure}%
    \caption{The figure (a) shows the planar map for $G$ given in
    \Cref{fig:drawing-graph} (b) with three different graph
    extensions: a path addition of $p$, a cyclic addition of $q$, and
    a spike addition of $r$. The additions of $p$ and $q$
    replace/divide the faces $F_2$ and $F_3$, we have in
    \Cref{fig:drawing-graph}, by two new faces for each addition. The
    addition of $r$ replaces $F_4$ with a face of a greater degree.
    The figure (b) shows the path addition discussed in the proof of
    \Cref{lem:face-division}.}
    \label{fig:planar-additions}
\end{figure}

\begin{lemma}\label{lem:path addition-preserves-connectedness} If $G$
is connected, then $G\bullet p$ and $G\bullet \overline{p}$ are
connected.
\end{lemma}

From now on, we will assume the existence of a planar map
\(\mathcal{M}\) of \(G\). We will only consider the addition of \(p\)
to \(G\) in a fixed face \(\mathcal{F}\) of \(\mathcal{M}\) if the
endpoints of \(p\), \(u\) and \(v\), belong to the boundary walk of
\(\mathcal{F}\). Under these considerations, one can prove that the
addition of \(p\) to \(G\) has a planar map.

\begin{lemma}\label{lem:face-division} There exists an extended planar
map of $\mathcal{M}$ for $G\bullet p$.
\end{lemma}

Given a planar map \(\mathcal{M}\) for \(G\), we label
\(E(\mathcal{M}, p)\) for the map given by \Cref{lem:face-division}.
The extended map \(E(\mathcal{M}, p)\) is called the
\emph{face division} of the face \(\mathcal{F}\) by \(p\), assuming
that \(p\) is embedded in \(\mathcal{F}\). To outline, the proof of
\Cref{lem:face-division} contains the following stages. First, one
should define a map that extends \(\mathcal{M}\) for the nodes in
\(p\). Second, as illustrated in \Cref{fig:planar-additions}, one
should define two faces, both induced by collocating \(p\) on
\(\mathcal{F}\). Finally, considering the new walks in
\(U(G \bullet p)\), given by walk-compositions with \(p\), we can
prove that \(E(\mathcal{M}, p)\) is planar.

\begin{proof}[Proof of \Cref{lem:face-division}]
For brevity, let $H$ be the graph $G \bullet p$ as constructed above,
and $p$ be the walk $\hat{u}\cdot p' \cdot \hat{v}$; a walk from $u$
to $v$. Let $\mathcal{F}$ be a face such that its boundary contains
$u$ and $v$. We will define a specific map $\mathcal{M'}$ for $H$ that
extends the given planar map $\mathcal{M}$ of $G$. In this way, one
embeds $p$ in $\mathcal{F}$. By \Cref{def:face}, in
$\partial\mathcal{F}$, one has the previous edge at $u$, denoted by
$a:\Edge_{G}(\mathsf{pred}(u),u)$, and an edge after $a$, denoted by
$a^{+}:\Edge_{G}(u, \mathsf{succ}(u))$. Similarly, for $v$, we have
$b:\Edge_{G}(\mathsf{pred}(v),v)$ and $b^{+}:\Edge_{G}(v,
\mathsf{succ}(v))$, as illustrated by \Cref{fig:planar-additions}b.

If $x = u$ then $\mathcal{M'}(x)$ is the cycle $\mathcal{M}(x)$ with
the insertion of $\hat{u}: \Edge_{H}(u, 0)$ between $a$ and $a^{+}$,
i.e. $\mathcal{M'}(u)$ is $(\cdots a\,\hat{u}\,a^{+}\cdots)$.
Similarly, if $x = v$, the cycle $\mathcal{M'}(v)$ is $(\cdots
b\,\hat{v}\,b^{+}\cdots)$. If $x$ is a node in $p'$, i.e. $x = i$ for
$i$ from $0$ to $n-1$, then $\mathcal{M'}(i)$ is $(e_{i}e_{i+1})$,
where $e_i:\Edge_{p'}(i,i+1)$. Otherwise, $x$ is in $G$, and
$\mathcal{M'}(x)$ is $\mathcal{M}(x)$. 

The path $p$ splits $\mathcal{F}$ in two faces, $F_1$ and $F_2$. Let
$\mathcal{F}$ be given by $(A, h)$ of degree $m$, as in
\Cref{def:face}, and $k$ be the length of the walk
$\mathsf{cw}_{\mathcal{F}}(u,v)$ from $u$ to $v$ in
$\partial\mathcal{F}$. For brevity, let $n_1,n_2$ be $k+(n+1)$ and
$(m-k)+(n+1)$, respectively. Let $F_1, F_2$ be the faces $(C_{n_1},
h_1)$ and $(C_{n_2}, h_2)$, where $h_1$ and $h_2$ are of type
$\mathsf{Hom}(C_{n_i}, U(H))$ for $i=1,2$. We will define $h_1$ and
$h_2$ in a way that their boundary walks are
$\mathsf{cw}_{\mathcal{F}}(u,v)\cdot U(p)$ and
$\mathsf{ccw}_{\mathcal{F}}(u,v)\cdot U(p)$, respectively.

Let $h_1$ be $(\alpha_1,\beta_1)$. If $i<k$ for $i : \Node_{C_{n_1}}$,
one puts the node $i$ in $\mathsf{cw}_{\mathcal{F}}(u,v)$, i.e.
$\alpha_1(i)$ is $\alpha(i)$ and consequently, $\beta_1(i,i+1,e)$ is
$\beta(i,i+1,e)$ for $e :\Edge_{C_{n_1}}(i,i+1)$. Otherwise, if $k\leq
i \leq n_1$, then one puts the node $i$ in $U(p)$, i.e. $\alpha_1(i)$
is $n-i$, and consequently, $\beta_1(i,i+1,e)$ is the edge
$\mathsf{inl}(e_{i})$ in $U(H)$. It is clear that $h_1$, and
similarly, $h_2$, is an edge-injective and map-compatible graph
homomorphism with $\mathcal{M'}$, since their properties are inherited from
$h$ and $p$.

Let us now prove that $\mathcal{M'}$ is spherical. Since this is a
proposition by \Cref{lem:spherical-is-proposition}, we can use the
elimination of the propositional truncation applied to the evidence
that $\mathcal{M}$ is spherical. This permits us to freely obtain a
walk homotopy for any pair of walks sharing endpoints in $U(G)$. On
the other hand, one must observe that a pair of homotopic walks in
$U(G)$ that deforms along different faces than $\mathcal{F}$, remain
homotopic in $U(H)$. Hence, we only need to consider for our goal, (i)
the set of walks in $U(G)$ deforming along $\mathcal{F}$ and (ii) the
set of walks created by compositions of $p$ with existing walks in
$U(G)$. 

If both walks belong to the former set, then we know that they are
homotopic by vertical composition of the homotopies along $F_1$ and
$F_2$ \cite[§5]{homotopywalks}.

\begin{figure}[!ht]
  \centering
  \[\begin{tikzcd}[row sep=large]
    {\bullet_{x}} & {\bullet_{u}} & \textcolor{darkblue}{\bullet_{y}} & {\bullet_{v}} & {\bullet_{z}}
    \arrow["{\delta_1}", no head, from=1-1, to=1-2]
    \arrow[""{name=0, anchor=center, inner sep=0}, "{p_2}", color={darkblue}, no head, from=1-3, to=1-4]
    \arrow[""{name=1, anchor=center, inner sep=0}, "{\mathsf{cw}_{\mathcal{F}}(u,v)}", curve={height=-30pt}, no head, from=1-2, to=1-4]
    \arrow["{\delta_2}", no head, from=1-4, to=1-5]
    \arrow[""{name=2, anchor=center, inner sep=0}, "{p_1}", color={darkblue}, no head, from=1-2, to=1-3]
    \arrow[""{name=3, anchor=center, inner sep=0}, "{\mathsf{ccw}_{\mathcal{F}}(u,v)}"', curve={height=30pt}, no head, from=1-2, to=1-4]
    \arrow["\,{h_{F_2}}"', shorten <=3pt, Rightarrow, 2tail reversed, no head, from=3, to=1-3]
    \arrow["\,{h_{F_1}}", shorten <=3pt, Rightarrow, 2tail reversed, no head, from=1, to=1-3]
  \end{tikzcd}\] \caption{The figure shows a part of the graph
  $U(G\bullet p)$ embedded in the sphere. As constructed in the proof
  of \Cref{lem:face-division}, the faces, $F_1$ and $F_2$, of the map
  $\mathcal{M'}$ are given by a face division of $\mathcal{F}$ by the
  path $p$. Such gives rise to new walk homotopies, as $h_{F_1}$ and
  $h_{F_2}$ in the picture.  The walk $U(p)$ from $u$ to $v$ is the
  walk composition of $p_1$, a walk from $u$ to $y$, and $p_2$, a walk
  from $y$ to $v$. The walks $\delta_1$ and $\delta_2$ are
  walks in $U(G)$ from $x$ to $z$.}
  \label{fig:lemma-path addition}
\end{figure}
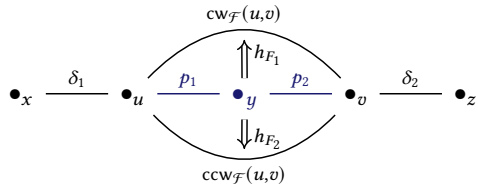

Otherwise, if both walks, $w_1$ and $w_2$ from $x$ to $z$, belong to
the latter set, then we proceed by case analysis to show that
$w_1\sim_{\mathcal{M'}} w_2$. For $w_1$ and $w_2$, it is only required
to consider walks without inner loops, by Lemma $5.8$ in
\cite{homotopywalks}. In the following, the variables $p_1,p_2,
\delta_1$, and $\delta_2$ denote walks, as in \Cref{fig:lemma-path
addition}.

\begin{itemize}
\item If $w_1$ is $\delta_1 \cdot U(p) \cdot \delta_2$, and $p$ is not
a subwalk of $w_2$, then one can obtain the following walk homotopy.
\begin{equation}\label{eq:p-reduces-w2}
  \begin{array}{ll}
      w_1 \equiv \delta_1 \cdot U(p) \cdot \delta_2 & \\
      {\color{white} w_1} \equiv \delta_1 \cdot \mathsf{ccw}_{F_1}(u,v) \cdot \delta_2  &(\mbox{By the construction of $F_1$}) \\
      {\color{white} w_1} \sim_{\mathcal{M'}} \delta_1 \cdot \mathsf{cw}_{F_1}(u,v) \cdot \delta_2  &(\mbox{By the constructor $\mathsf{hcollapse}$ with $F_1$, $\delta_1$, and $\delta_2$}) \\
      {\color{white} w_1} \equiv \delta_1 \cdot \mathsf{cw}_{\mathcal{F}}(u,v) \cdot \delta_2  &(\mbox{By the construction of $F_1$}) \\
      {\color{white} w_1} \sim_\mathcal{M} w_2 &(\mbox{By the spherical map $\mathcal{M}$ applied to walks from $x$ to $z$}).
  \end{array}
\end{equation}
\item If $w_1$ is $\delta_1 \cdot p_1$ and $w_2$ is $\delta \cdot p_1$
for a walk $\delta$ from $x$ to $u$, then, by right whiskering of walk
homotopies, one gets that $\delta_1 \cdot p_1 \sim_{\mathcal{M'}}
\delta \cdot p_1$. By assumption, $\mathcal{M}$ is spherical, and then
$\delta_1 \sim_{\mathcal{M}} \delta$, which implies by definition of
$\mathcal{M'}$ that $\delta_1 \sim_{\mathcal{M'}} \delta$. Similarly,
by left whiskering, one can also prove that if $w_1$ is $p_2\cdot
\delta_2$ and $w_2$ is $p_2 \cdot \delta$ where, $\delta$ is a walk
from $v$ to $z$, then there is a walk homotopy such that $p_2\cdot
\delta_2 \sim_{\mathcal{M'}} p_2 \cdot \delta$.
\end{itemize}

For the remainder cases of $w_1$ and $w_2$, one can similarly
construct the required walk homotopy. Therefore, $\mathcal{M'}$ is
spherical, and is also a planar map of $H$ with the outer face $F_1$.
\end{proof}

If \(G\) is a finite graph with a map \(\mathcal{M}\), then the
Euler's characteristic of \(G\) by \(\mathcal{M}\), denoted by
\(\chi_{\mathcal{M}}\), is the number relating the cardinal of the set
of nodes (\(n\)), edges (\(e\)), and faces (\(f\)).

\begin{equation}\label{eq:euler}
\chi_{\mathcal{M}} :\equiv n - e + f.
\end{equation}

One can prove that if \(\mathcal{M}\) is a planar map then
\(\chi_{\mathcal{M}}\) and \(\chi_{E(\mathcal{M},G)}\) are equal. As
described in the proof of \Cref{lem:face-division} to construct
\(E(\mathcal{M},G)\), a path addition of \(p\) of length \(k+1\) to
\(G\) increases \(n\) by \(k\), \(e\) by \(k+1\), and \(f\) by one. A
major result is the characterisation of connected and finite planar
graphs by Euler's formula, which states that \(\chi_{\mathcal{M}}\) is
two. Using the development in this section, one could show Euler's
formula for graph extensions, and for the class of biconnected graphs
as described in \Cref{sec:biconnected-graph}. However, it is still
unclear how to verify Euler's formula, when it is not given the
cardinal of the set of faces, i.e.~elements of type
\(\mathsf{Face}(G,\mathcal{M})\), see \Cref{def:face}.

\begin{equation*}
\chi_{E(\mathcal{M},p)} :\equiv (n + k) - (e + k + 1) + (f + 1) = \chi_{\mathcal{M}}.
\end{equation*}

\begin{figure}[!htb]
  \centering
  \includegraphics[width=0.68\textwidth]{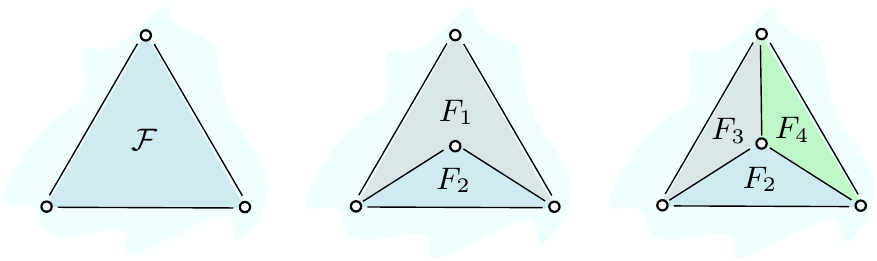}
  \caption{The figure is a planar synthesis of the construction of a
  planar map for $K_4$ from a planar map of $C_3$. One first divides
  the face $\mathcal{F}$ into $F_1$ and $F_2$. Then one splits $F_1$
  into $F_3$ and $F_4$.}
  \label{fig:k4-is-planar}
\end{figure}

There are several methods to construct graphs inductively, as the
construction of \(K_4\) in \Cref{fig:k4-is-planar}. Whitney-Robbins
synthesis and an ear decomposition of a graph are some related
methods. Inspired by these constructions and \Cref{lem:face-division},
we define the construction of larger planar graphs using graph
extensions, in a way that we never leave the class of planar graphs.

\begin{definition}\label{def:whitney-synthesis} A \emph{synthesis} of
  a graph $G$ from a graph $H$ is a sequence of graphs $G_0,G_1,
  \cdots, G_n$ where $G_0$ is $H$, $G_n$ is $G$, and $G_i$ is the
  addition of $p_{i}$ to $G_{i-1}$ for $i$ from $1$ to $n$. If the
  sequence only contains simple additions, then it is called a
  \emph{simple synthesis}. Else, if the sequence only contains
  non-simple additions, the sequence is called a \emph{non-simple
  synthesis}. If the graph $G$ is only obtained by a sequence of path
  additions, then the sequence is called a \emph{Whitney synthesis}. 
\end{definition}

\begin{lemma} \label{lem:whitney-is-connected} In a synthesis from a
connected graph, every graph in the sequence is connected.
\end{lemma}

\begin{definition}\label{def:planar-synthesis} Given a planar map
  $\mathcal{M}$ of the graph $G$, a \emph{planar synthesis} of $H$
  from a graph $G$ is a sequence $(G_0, \mathcal{M}_0), (G_1,
  \mathcal{M}_1) \cdots, (G_n, \mathcal{M}_n),$ where $n$ represents
  the length of the synthesis, $(G_0, \mathcal{M}_0)$ is $(G ,
  \mathcal{M})$, and $(G_n, \mathcal{M}_{n})$ is $(H,
  E(\mathcal{M}_{n-1},p_{n-1}))$. The graph $G_i$, for $i$ from $1$ to
  $n$, is the graph $G \bullet p_i$ and the map $\mathcal{M}_i$ is
  $E(\mathcal{M}_{i-1},p_{i-1})$, for $i$ from $1$ to $n$.
\end{definition}

\begin{lemma}\label{lem:planar-synthesis-have-a-planar-map} In a
planar synthesis, every graph in the sequence is planar.
\end{lemma}

\begin{proof}
By induction on the synthesis length and \Cref{lem:face-division}.
\end{proof}

\Cref{lem:face-division} can be further extended to consider
non-simple additions, and consequently, one could extend
\Cref{lem:planar-synthesis-have-a-planar-map} to define
\emph{non-simple planar syntheses}. Given a map \(\mathcal{M}\) for
\(G\), the corresponding planar map for \(G \bullet \overline{p}\) is
denoted by \(E(\mathcal{M},\overline{p})\). Similarly, as with path
additions, by extending the map by non-simple additions, new faces
show up. Let \(k+1\) be the length of the path added to \(G\). Then
the map \(E(\mathcal{M},\overline{p})\) induce \(k+2\) new faces, and
one gets that \(\chi_{E(\mathcal{M},\overline{p})}\) and
\(\chi_{\mathcal{M}}\) are equal.

\begin{equation*}
  \chi_{E(\mathcal{M},\overline{p})} :\equiv (n + k) - (e + 2\cdot(k + 1)) + (f+(k+2)) = \chi_{\mathcal{M}}.
\end{equation*}

As illustrated by \Cref{fig:planar-additions}, using spike additions,
larger planar graphs can be constructed. A \emph{spike addition} to
\(G\) is the addition of a path that only has one node in common with
\(G\). Given a map for \(G\), a simple addition of a spike \(p\) to
\(G\) induces a new face of a greater degree than the face where the
spike is inserted. In contrast to simple additions, the number of
faces of a map extended by non-simple spike additions vary as new
faces arise between pairs of edges that share their endpoints.

\subsubsection{Biconnected Planar Graphs}\label{sec:biconnected-graph}

One can look at how much a graph is connected by examining its
node-connectivity or edge-connectivity. Some graphs, yet after
removing parts of them, preserve one or both connectivity measures. In
this subsection, we want to characterise how to construct the class of
\(2\)-connected planar graphs. A graph is \(k\)-\emph{connected} if it
cannot be disconnected by removing less than \(k\) nodes. There exist,
depending on \(k\), different ways to construct the class of
\(k\)-connected graphs. For example, it is known that one can
construct any undirected (2)-connected graph, if one applies path
additions to a proper cyclic graph \citep[§3]{diestel}.

\begin{definition}\label{def:biconnected-graph} A graph $G$ is
  $2$-\emph{connected} or biconnected if the graph formed by removing
  from $G$ a node $x$, denoted by $G - x$, is connected. If $G$ is
  $2$-connected then the proposition $\mathsf{Biconnected}(G)$ holds.
  Precisely, $G - x$ is the graph formed by the set of nodes,
  $\Sigma_{y : \Node_{G}} (x \neq y)$, and the corresponding edges in
  $G$.
\begin{equation*}
\mathsf{Biconnected}(G):\equiv \prod_{x:\Node_{G}}\, \mathsf{Connected}(G - x).
\end{equation*}
\end{definition}

\begin{lemma}\label{lem:u-cyclic-is-biconnected}
  If $G$ is a cyclic graph, then $U(G)$ is $2$-connected.
\end{lemma}

The \(2\)-connectedness of a graph is not preserved by simple path
additions. Clearly, removing a node from the added path \(p\)
disconnects \(G \bullet p\). However, using non-simple path additions,
we can preserve and enlarge \(2\)-connected graphs.

\begin{lemma}\label{lem:2-connected}
Suppose $G$ is a $2$-connected graph, then the following claims hold.
\begin{enumerate}
  \item Every node in $G$ has degree of minimum two.
  \item There exists a cyclic graph $H$ and an injective morphism from
  $U(H)$ to $G$.
  \item The graphs $G \bullet \overline{p}$, $U(G \bullet p)$, and
  $U(G) \bullet \overline{p}$ are all $2$-connected.
\end{enumerate}
\end{lemma}

\begin{lemma}\label{lem:2-connected-planar-in}
In a non-simple Whitney synthesis of $G$ of length $n$ from a
$2$-connected cyclic graph $H$, every graph $G_i$ in the sequence is a
$2$-connected planar graph.
\end{lemma}

\begin{proof} By induction on $n$. If $n = 0$, the graph $H$ is
$2$-connected by hypothesis and is planar by a similar construction as
in \Cref{lem:Cn-is-planar}. Assuming that the claim holds for a
sequence of length $n$, then $G_{n}$ is a $2$-connected planar graph.
By Item (iii) in \Cref{lem:2-connected}, we get that $G_{n} \bullet
\overline{p_i}$ is $2$-connected. Using a similar construction as in
the proof of \Cref{lem:face-division}, one defines a planar map for
$G_{n} \bullet \overline{p_i}$, from where the conclusion follows.
\end{proof}

The converse of \Cref{lem:2-connected-planar-in} can be proved by
closely following the informal proof of Lemma \(3\) and Proposition
\(4\) for undirected \(2\)-connected planar graphs in
\citep{yamamoto}. We must formalise several notions before considering
such a proof in our formalism, including, the notion of maximal sub
graphs, adjacent faces, and deletion of edge sequences. One
understands, therefore, that the class of \(2\)-connected planar
graphs is completely determined by all non-simple Whitney syntheses
\citep[§3]{diestel}. Any planar graph, in the sense of
\Cref{def:planar-graph}, and \(2\)-connected, as in
\Cref{def:biconnected-graph}, can be inductively generated from a
cycle graph and iterative additions of non-simple paths.

Further investigation to study of other graph extensions to generate
planar graphs, as graph amalgamations, graph appendages, deletions,
contractions and subdivisions should be considered
\citep[§7.3]{Gross2018}.

\section{Related Work}\label{sec:related-work}

One can find the study of planar graphs and more general graph
theoretic topics in relevant projects and big libraries formalised in
Coq \citep{doczkal2020} and Isabelle/HOL \citep{Noschinski2015}. For
example, the formal proof of the Four-Colour Theorem (FCT) in Coq by
Gonthier \citep{Gonthier2008}, the proof of the discrete form of the
Jordan Curve Theorem in Coq by Dufourd \citep{Dufourd2000}, and the
proof of the Kepler's Conjecture in HOL by Bauer et al.
\citep{hales-kepler} are a few of such notable projects in the
subject.

Different approaches have been proposed to address planarity of graphs
in formal systems. These works use different mathematical objects
depending on the system. We use combinatorial maps in this work, but
other related constructions are, for example, root maps defined in
terms of permutations by Dubois et al. \citep{CatherineDubois2016},
and hypermaps by Dufourd and Gonthier
\citep{Dufourd2000, dufourd2009, Gonthier2008}, among others. In
particular, one can see that the notion of a \emph{hypermap} is a
generalisation of a combinatorial map for undirected finite graphs.
Such a concept is one fundamental construction to formalise
mathematics of graph embeddings amongst in theorem provers, along with
the computer-checked proof of FCT. Additionally, Dufourd states and
proves the Euler's polyhedral formula and the Jordan Curve Theorem
using an inductive characterisation of hypermaps
\citep{Dufourd2000, dufourd2009}. Recently, for a more standard
representation of finite graphs, Doczkal proved that, according to his
notion of a plane map based on hypermaps, every \(K_{3,3}\)-free graph
and \(K_{5}\)-free graph without isolating vertices is planar, a
direction of Wagner's theorem \citep{doczkal2021}.

An alternative approach for planarity using combinatorial maps is the
iterative construction of certain kind of planar graphs. For example,
Yamamoto et al. \citep{yamamoto} showed that every biconnected and
finite planar graph can be decomposed as a finite set of cycle graphs,
where every face is the region bounded by a closed walk \citep[§5.2,
§7.3]{Gross2018}. Such construction defines an inductive data type
that begins with a cycle graph \(C_n\) serving as the base case, and
by repeatedly merging new instances of cycle graphs, one gets the
final planar graph. Bauer formalises in Isabelle/HOL a similar
construction of planar graphs from a set of faces
\citep{Bauer, BauerN02}. A related approach in our setting is of the
treatment of planar graph extensions, as described in
\Cref{sec:planar-extensions}.

However, to the best of our knowledge, in type theory, related work to
the planarity of graphs has been done in a different formal system and
for different classes of graphs. These studies mostly define planarity
for undirected finite graphs, in contrast, our definition considers
the more general class of connected and locally finite directed
multigraphs. Our work is closer to the foundations of mathematics,
specifically, to the formalisation of mathematics in HoTT, than to
more practical aspects of graph theory. This approach forces us to
propose new constructions, even sometimes, for the most fundamental
and basic concepts in the theory.

\section{Concluding Remarks}\label{sec:conclusions}

This document is a case study of graph-theoretic concepts in
constructive mathematics using homotopy type theory. An elementary
characterisation of planarity of connected and locally finite directed
multigraphs is presented in \Cref{sec:type-of-planar-maps}. We
collected all the maps of a graph in the two-dimensional plane
---identified up to isotopy--- in a homotopy set, see
\Cref{thm:being-planar-is-set}. The type of these planar maps displays
some of our main contributions, e.g.~the type of spherical maps stated
in \Cref{def:combinatorial-map} and the type of faces for a given map
in \Cref{def:face}. As far as we know, the presentation of these types
in a dependent type theory like HoTT is novel. For example, besides
its rather technical definition, we believe the type of faces encodes
in a better combinatorial way the essence of the topological intuition
behind it, rather than, being defined as simply cyclic lists of nodes,
as by other authors \citep{Gonthier2008, yamamoto, Bauer}, see
\Cref{sec:type-of-faces}.

Additionally, as a way to construct planar graphs inductively, we
presented extensions for planar maps. We demonstrated that any cycle
graph is planar, and by means of planar extensions like path
additions, one can construct larger planar graphs, e.g.~to illustrate
this approach, a planar map for \(K_4\) using simply path additions
from a planar map of \(C_3\) is illustrated in
\Cref{fig:k4-is-planar}. Other relevant notions to this work are
cyclic types, cyclic graphs, homotopy for walks \citep{homotopywalks},
and spherical maps.

We chose HoTT as the reasoning framework to directly study the
symmetry of our mathematical constructions. Many of the proofs
supporting our development could only be constructed by adopting the
Univalence Axiom, a main principle in HoTT. A primary example of using
Univalence in this paper is the structural identity principle for
graphs, as stated in \Cref{thm:equivalence-principle}.

Another contribution of this work include the (computer-checked)
proofs. The major results in this document have been formalised in the
proof assistant Agda, in a development fully self-contained way, which
does not depend on any library \citep{agdaformalisation}. However, for
technical reasons, the formalisation of \Cref{lem:Cn-is-planar} and
further in depth studies on the main results in
\Cref{sec:planar-extensions} like
\Cref{lem:face-division,lem:2-connected-planar-in} will be conducted
in future.

This work can serve as a starting point for further developments of
graph theory in HoTT or related dependent type theories. We expect
further research to provide other interesting results as the
equivalences between different characterisations of planarity for
graphs, e.g., the Kuratowski's and Wagner's characterisations for
planar graphs.

\begin{acks}
We thank the organisers of the conference TYPES2019, especially
Marc Bezem and the referees, for their constructive criticism on a
preliminary presentation of this work. The first author wants to
thank Noam Zeilberger for interesting discussions related to this
work and the references on rooted maps and planarity criteria.
Thanks to the organisers of Midlands Graduate School 2019.
Finally, for proofreading earlier versions of this work, the first
author thanks Benjamin Chetioui and Tam Thanh Truong. Thanks to
the Agda developer team for providing and maintaining the proof
assistant used in this work.
\end{acks}
\bibliography{ref.bib}
\end{document}

%% file: macros/names.tex
\newcommand{\eout}[1]{{#1}^\to}
\newcommand{\ein}[1]{{#1}^\leftarrow}

\newcommand{\Aut}{\mathsf{Aut}}

\newcommand{\coe}[1]{\ensuremath{\mathsf{coe}}\left (#1\right)}

\newcommand{\Connected}{\ensuremath{\mathsf{Connected}}}

\newcommand{\Edge}{\ensuremath{\mathsf{E}}}

\newcommand{\Face}[1]{\ensuremath{\mathsf{Face}({#1})}}

\newcommand{\fpred}[1]{\ensuremath{\mathsf{pred}(#1)}}

\newcommand{\Graph}{\ensuremath{\mathsf{Graph}}}

\newcommand{\HomWalk}[3]{\ensuremath{#2\sim_{#1}#3}}
\newcommand{\Hom}[2]{\ensuremath{\mathsf{Hom}({#1},{#2})}}

\newcommand{\hProp}{\ensuremath{\mathsf{hProp}}}

\newcommand{\isSet}[1]{\ensuremath{\mathsf{isSet}(#1)}}

\newcommand{\Map}[1]{\ensuremath{\mathsf{Map}(#1)}}

\newcommand{\Node}{\ensuremath{\mathsf{N}}}
\newcommand{\N}{\mathbb{N}}

\newcommand{\trtwo}[4]{\ensuremath{\mathsf{tr_2}^{#1}(#2,#3,#4)}}

\newcommand{\trunc}[1]{\|#1\|}
\newcommand{\tr}[3]{\ensuremath{\mathsf{tr}^{#1}(#2,#3)}}

\newcommand{\UU}{\mathcal{U}}


%% file: planar/abstract.md
\begin{abstract}

In this paper, we present a constructive and proof-relevant
development of graph theory, including the notion of maps, their
faces, and maps of graphs embedded in the sphere, in homotopy type
theory. This allows us to provide an elementary characterisation of
planarity for locally directed finite and connected multigraphs that
takes inspiration from topological graph theory, particularly from
combinatorial embeddings of graphs into surfaces. A graph is planar if
it has a map and an outer face with which any walk in the embedded
graph is walk-homotopic to another. A result is that this type of
planar maps forms a homotopy set for a graph. As a way to construct
examples of planar graphs inductively, extensions of planar maps are
introduced. We formalise the essential parts of this work in the
proof assistant Agda with support for homotopy type theory.

\end{abstract}

%% file: planar.bbl

\begin{thebibliography}{43}


\ifx \showCODEN    \undefined \def \showCODEN     #1{\unskip}     \fi
\ifx \showDOI      \undefined \def \showDOI       #1{#1}\fi
\ifx \showISBNx    \undefined \def \showISBNx     #1{\unskip}     \fi
\ifx \showISBNxiii \undefined \def \showISBNxiii  #1{\unskip}     \fi
\ifx \showISSN     \undefined \def \showISSN      #1{\unskip}     \fi
\ifx \showLCCN     \undefined \def \showLCCN      #1{\unskip}     \fi
\ifx \shownote     \undefined \def \shownote      #1{#1}          \fi
\ifx \showarticletitle \undefined \def \showarticletitle #1{#1}   \fi
\ifx \showURL      \undefined \def \showURL       {\relax}        \fi
\providecommand\bibfield[2]{#2}
\providecommand\bibinfo[2]{#2}
\providecommand\natexlab[1]{#1}
\providecommand\showeprint[2][]{arXiv:#2}

\bibitem[\protect\citeauthoryear{Ahrens and North}{Ahrens and North}{2019}]%
        {ahrens2018}
\bibfield{author}{\bibinfo{person}{Benedikt Ahrens} {and}
  \bibinfo{person}{Paige~Randall North}.} \bibinfo{year}{2019}\natexlab{}.
\newblock \bibinfo{booktitle}{\emph{Univalent Foundations and the Equivalence
  Principle}}.
\newblock \bibinfo{publisher}{Springer International Publishing},
  \bibinfo{address}{Cham}, \bibinfo{pages}{137--150}.
\newblock
\urldef\tempurl%
\url{https://doi.org/10.1007/978-3-030-15655-8\_6}
\showDOI{\tempurl}


\bibitem[\protect\citeauthoryear{Ahrens, North, Shulman, and Tsementzis}{Ahrens
  et~al\mbox{.}}{2020}]%
        {Ahrens2020}
\bibfield{author}{\bibinfo{person}{Benedikt Ahrens},
  \bibinfo{person}{Paige~Randall North}, \bibinfo{person}{Michael Shulman},
  {and} \bibinfo{person}{Dimitris Tsementzis}.}
  \bibinfo{year}{2020}\natexlab{}.
\newblock \showarticletitle{A Higher Structure Identity Principle}. In
  \bibinfo{booktitle}{\emph{{LICS} '20: 35th Annual {ACM/IEEE} Symposium on
  Logic in Computer Science, Saarbr{\"{u}}cken, Germany, July 8-11, 2020}},
  \bibfield{editor}{\bibinfo{person}{Holger Hermanns}, \bibinfo{person}{Lijun
  Zhang}, \bibinfo{person}{Naoki Kobayashi}, {and} \bibinfo{person}{Dale
  Miller}} (Eds.). \bibinfo{publisher}{Acm},
  \bibinfo{address}{Saarbr{\"{u}}cken, Germany}, \bibinfo{pages}{53--66}.
\newblock
\urldef\tempurl%
\url{https://doi.org/10.1145/3373718.3394755}
\showDOI{\tempurl}


\bibitem[\protect\citeauthoryear{Appel and Haken}{Appel and Haken}{1986}]%
        {Appel1986}
\bibfield{author}{\bibinfo{person}{K. Appel} {and} \bibinfo{person}{W. Haken}.}
  \bibinfo{year}{1986}\natexlab{}.
\newblock \showarticletitle{The four color proof suffices}.
\newblock \bibinfo{journal}{\emph{The Mathematical Intelligencer}}
  \bibinfo{volume}{8}, \bibinfo{number}{1} (\bibinfo{date}{March}
  \bibinfo{year}{1986}), \bibinfo{pages}{10--20}.
\newblock
\urldef\tempurl%
\url{https://doi.org/10.1007/bf03023914}
\showDOI{\tempurl}


\bibitem[\protect\citeauthoryear{Archdeacon}{Archdeacon}{1996}]%
        {surveytopgraph}
\bibfield{author}{\bibinfo{person}{Dan Archdeacon}.}
  \bibinfo{year}{1996}\natexlab{}.
\newblock \showarticletitle{Topological Graph Theory -- A Survey}.
\newblock \bibinfo{journal}{\emph{Cong. Num}}  \bibinfo{volume}{115}
  (\bibinfo{year}{1996}), \bibinfo{pages}{115--5}.
\newblock


\bibitem[\protect\citeauthoryear{Avigad and Harrison}{Avigad and
  Harrison}{2014}]%
        {avigad-fvm}
\bibfield{author}{\bibinfo{person}{Jeremy Avigad} {and} \bibinfo{person}{John
  Harrison}.} \bibinfo{year}{2014}\natexlab{}.
\newblock \showarticletitle{Formally Verified Mathematics}.
\newblock \bibinfo{journal}{\emph{Commun. ACM}} \bibinfo{volume}{57},
  \bibinfo{number}{4} (\bibinfo{year}{2014}), \bibinfo{pages}{66--75}.
\newblock


\bibitem[\protect\citeauthoryear{Awodey}{Awodey}{2012}]%
        {Awodey2012}
\bibfield{author}{\bibinfo{person}{Steve Awodey}.}
  \bibinfo{year}{2012}\natexlab{}.
\newblock \showarticletitle{Type Theory and Homotopy}.
\newblock In \bibinfo{booktitle}{\emph{Epistemology versus Ontology}}.
  \bibinfo{publisher}{Springer Netherlands}, \bibinfo{address}{Pitt, Usa},
  \bibinfo{pages}{183--201}.
\newblock
\urldef\tempurl%
\url{https://doi.org/10.1007/978-94-007-4435-6\_9}
\showDOI{\tempurl}


\bibitem[\protect\citeauthoryear{Awodey}{Awodey}{2018}]%
        {Awodey2018}
\bibfield{author}{\bibinfo{person}{Steve Awodey}.}
  \bibinfo{year}{2018}\natexlab{}.
\newblock \showarticletitle{Univalence as a principle of logic}.
\newblock \bibinfo{journal}{\emph{Indagationes Mathematicae}}
  \bibinfo{volume}{29}, \bibinfo{number}{6} (\bibinfo{date}{Dec.}
  \bibinfo{year}{2018}), \bibinfo{pages}{1497--1510}.
\newblock
\urldef\tempurl%
\url{https://doi.org/10.1016/j.indag.2018.01.011}
\showDOI{\tempurl}


\bibitem[\protect\citeauthoryear{Baez, Hoffnung, and Walker}{Baez
  et~al\mbox{.}}{9 08}]%
        {Baez2009}
\bibfield{author}{\bibinfo{person}{John~C. Baez}, \bibinfo{person}{Alexander~E.
  Hoffnung}, {and} \bibinfo{person}{Christopher~D. Walker}.}
  \bibinfo{year}{2009-08}\natexlab{}.
\newblock \showarticletitle{Higher-Dimensional Algebra VII: Groupoidification}.
\newblock \bibinfo{journal}{\emph{Theory and Applications of Categories}}
  \bibinfo{volume}{24} (\bibinfo{year}{2009-08}), \bibinfo{pages}{489--553}.
\newblock
\urldef\tempurl%
\url{http://arxiv.org/abs/0908.4305}
\showURL{%
\tempurl}


\bibitem[\protect\citeauthoryear{Bang-Jensen and Gutin}{Bang-Jensen and
  Gutin}{2009}]%
        {BangJensen2009}
\bibfield{author}{\bibinfo{person}{J{\o}rgen Bang-Jensen} {and}
  \bibinfo{person}{Gregory~Z. Gutin}.} \bibinfo{year}{2009}\natexlab{}.
\newblock \bibinfo{booktitle}{\emph{Digraphs}}.
\newblock \bibinfo{publisher}{Springer London}, \bibinfo{address}{London, UK}.
\newblock
\urldef\tempurl%
\url{https://doi.org/10.1007/978-1-84800-998-1}
\showDOI{\tempurl}


\bibitem[\protect\citeauthoryear{Bauer and Nipkow}{Bauer and Nipkow}{2002}]%
        {BauerN02}
\bibfield{author}{\bibinfo{person}{Gertrud Bauer} {and} \bibinfo{person}{Tobias
  Nipkow}.} \bibinfo{year}{2002}\natexlab{}.
\newblock \showarticletitle{The 5 Colour Theorem in Isabelle/Isar}. In
  \bibinfo{booktitle}{\emph{Theorem Proving in Higher Order Logics, 15th
  International Conference, TPHOLs 2002, Hampton, VA, USA, August 20-23, 2002,
  Proceedings}}, \bibfield{editor}{\bibinfo{person}{Victor~A. Carre{\~{n}}o},
  \bibinfo{person}{C{\'e}sar~A. Mu{\~{n}}oz}, {and}
  \bibinfo{person}{Sofi{\`e}ne Tahar}} (Eds.). \bibinfo{publisher}{Springer
  Berlin Heidelberg}, \bibinfo{address}{Berlin, Heidelberg},
  \bibinfo{pages}{67--82}.
\newblock
\urldef\tempurl%
\url{https://doi.org/10.1007/3-540-45685-6\_6}
\showDOI{\tempurl}


\bibitem[\protect\citeauthoryear{Bauer}{Bauer}{2005}]%
        {Bauer}
\bibfield{author}{\bibinfo{person}{Gertrud~Josefine Bauer}.}
  \bibinfo{year}{2005}\natexlab{}.
\newblock \emph{\bibinfo{title}{Formalizing Plane Graph Theory: Towards a
  Formalized Proof of the Kepler Conjecture}}.
\newblock \bibinfo{thesistype}{Ph.D. Dissertation}. \bibinfo{school}{Technische
  Universit\"{a}t M\"{u}nchen}, \bibinfo{address}{Germany}.
\newblock
\urldef\tempurl%
\url{https://mediatum.ub.tum.de/doc/601794/document.pdf}
\showURL{%
\tempurl}


\bibitem[\protect\citeauthoryear{Baur}{Baur}{2012}]%
        {Baur2012}
\bibfield{author}{\bibinfo{person}{Melanie Baur}.}
  \bibinfo{year}{2012}\natexlab{}.
\newblock \emph{\bibinfo{title}{Combinatorial Concepts and Algorithms for
  Drawing Planar Graphs}}.
\newblock \bibinfo{thesistype}{Ph.D. Dissertation}.
  \bibinfo{school}{Universit\"{a}t Konstanz}, \bibinfo{address}{Konstanz}.
\newblock
\urldef\tempurl%
\url{http://nbn-resolving.de/urn:nbn:de:bsz:352-202281}
\showURL{%
\tempurl}


\bibitem[\protect\citeauthoryear{Cockx, Devriese, and Piessens}{Cockx
  et~al\mbox{.}}{2016}]%
        {COCKX2016}
\bibfield{author}{\bibinfo{person}{Jesper Cockx}, \bibinfo{person}{Dominique
  Devriese}, {and} \bibinfo{person}{Frank Piessens}.}
  \bibinfo{year}{2016}\natexlab{}.
\newblock \showarticletitle{Eliminating dependent pattern matching without K}.
\newblock \bibinfo{journal}{\emph{Journal of Functional Programming}}
  \bibinfo{volume}{26} (\bibinfo{year}{2016}), \bibinfo{pages}{e16}.
\newblock
\urldef\tempurl%
\url{https://doi.org/10.1017/s0956796816000174}
\showDOI{\tempurl}


\bibitem[\protect\citeauthoryear{Coquand and Danielsson}{Coquand and
  Danielsson}{2013}]%
        {iso-implies-equality}
\bibfield{author}{\bibinfo{person}{Thierry Coquand} {and}
  \bibinfo{person}{Nils~Anders Danielsson}.} \bibinfo{year}{2013}\natexlab{}.
\newblock \showarticletitle{Isomorphism is equality}.
\newblock \bibinfo{journal}{\emph{Indagationes Mathematicae}}
  \bibinfo{volume}{24}, \bibinfo{number}{4} (\bibinfo{year}{2013}),
  \bibinfo{pages}{1105--1120}.
\newblock
\showISSN{0019-3577}
\urldef\tempurl%
\url{https://doi.org/10.1016/j.indag.2013.09.002}
\showDOI{\tempurl}


\bibitem[\protect\citeauthoryear{Diestel}{Diestel}{2012}]%
        {diestel}
\bibfield{author}{\bibinfo{person}{Reinhard Diestel}.}
  \bibinfo{year}{2012}\natexlab{}.
\newblock \bibinfo{booktitle}{\emph{Graph Theory, 4th Edition}}.
  \bibinfo{series}{Graduate texts in mathematics}, Vol.~\bibinfo{volume}{173}.
\newblock \bibinfo{publisher}{Springer}, \bibinfo{address}{Hamburg, Germany}.
\newblock
\urldef\tempurl%
\url{https://doi.org/10.1007/978-3-662-53622-3}
\showDOI{\tempurl}


\bibitem[\protect\citeauthoryear{Doczkal}{Doczkal}{2021}]%
        {doczkal2021}
\bibfield{author}{\bibinfo{person}{Christian Doczkal}.}
  \bibinfo{year}{2021}\natexlab{}.
\newblock \bibinfo{title}{A Variant of Wagner's Theorem Based on Combinatorial
  Hypermaps}.  (\bibinfo{date}{Feb.} \bibinfo{year}{2021}).
\newblock
\urldef\tempurl%
\url{https://hal.archives-ouvertes.fr/hal-03142192}
\showURL{%
\tempurl}
\newblock
\shownote{working paper or preprint.}


\bibitem[\protect\citeauthoryear{Doczkal and Pous}{Doczkal and Pous}{2020}]%
        {doczkal2020}
\bibfield{author}{\bibinfo{person}{Christian Doczkal} {and}
  \bibinfo{person}{Damien Pous}.} \bibinfo{year}{2020}\natexlab{}.
\newblock \showarticletitle{Graph Theory in Coq: Minors, Treewidth, and
  Isomorphisms}.
\newblock \bibinfo{journal}{\emph{J. Autom. Reason.}} \bibinfo{volume}{64},
  \bibinfo{number}{5} (\bibinfo{year}{2020}), \bibinfo{pages}{795--825}.
\newblock
\urldef\tempurl%
\url{https://doi.org/10.1007/s10817-020-09543-2}
\showDOI{\tempurl}


\bibitem[\protect\citeauthoryear{Dubois, Giorgetti, and Genestier}{Dubois
  et~al\mbox{.}}{2016}]%
        {CatherineDubois2016}
\bibfield{author}{\bibinfo{person}{Catherine Dubois}, \bibinfo{person}{Alain
  Giorgetti}, {and} \bibinfo{person}{Richard Genestier}.}
  \bibinfo{year}{2016}\natexlab{}.
\newblock \showarticletitle{Tests and Proofs for Enumerative Combinatorics}. In
  \bibinfo{booktitle}{\emph{Tests and Proofs - 10th International Conference,
  TAP\@STAF 2016, Vienna, Austria, July 5-7, 2016, Proceedings}}
  \emph{(\bibinfo{series}{Lecture Notes in Computer Science})},
  \bibfield{editor}{\bibinfo{person}{Bernhard~K. Aichernig} {and}
  \bibinfo{person}{Carlo~A. Furia}} (Eds.), Vol.~\bibinfo{volume}{9762}.
  \bibinfo{publisher}{Springer}, \bibinfo{address}{Vienna, Austria},
  \bibinfo{pages}{57--75}.
\newblock
\urldef\tempurl%
\url{https://doi.org/10.1007/978-3-319-41135-4\_4}
\showDOI{\tempurl}


\bibitem[\protect\citeauthoryear{Dufourd}{Dufourd}{2009}]%
        {dufourd2009}
\bibfield{author}{\bibinfo{person}{Jean~Fran\c{c}ois Dufourd}.}
  \bibinfo{year}{2009}\natexlab{}.
\newblock \showarticletitle{An intuitionistic proof of a discrete form of the
  Jordan curve theorem formalized in Coq with combinatorial hypermaps}.
\newblock \bibinfo{journal}{\emph{Journal of Automated Reasoning}}
  \bibinfo{volume}{43}, \bibinfo{number}{1} (\bibinfo{year}{2009}),
  \bibinfo{pages}{19--51}.
\newblock
\showISSN{01687433}
\urldef\tempurl%
\url{https://doi.org/10.1007/s10817-009-9117-x}
\showDOI{\tempurl}


\bibitem[\protect\citeauthoryear{Dufourd and Puitg}{Dufourd and Puitg}{2000}]%
        {Dufourd2000}
\bibfield{author}{\bibinfo{person}{Jean-Fran\c{c}ois Dufourd} {and}
  \bibinfo{person}{Fran\c{c}ois Puitg}.} \bibinfo{year}{2000}\natexlab{}.
\newblock \showarticletitle{Functional specification and prototyping with
  oriented combinatorial maps}.
\newblock \bibinfo{journal}{\emph{Computational Geometry}}
  \bibinfo{volume}{16}, \bibinfo{number}{2} (\bibinfo{year}{2000}),
  \bibinfo{pages}{129--156}.
\newblock
\showISSN{0925-7721}
\urldef\tempurl%
\url{https://doi.org/10.1016/S0925-7721(00)00004-3}
\showDOI{\tempurl}


\bibitem[\protect\citeauthoryear{Escard\'{o}}{Escard\'{o}}{2018}]%
        {escardoUA}
\bibfield{author}{\bibinfo{person}{Mart\'{\i}n~H\"{o}tzel Escard\'{o}}.}
  \bibinfo{year}{2018}\natexlab{}.
\newblock \bibinfo{title}{A self-contained, brief and complete formulation of
  Voevodsky's Univalence Axiom}.
\newblock
\newblock
\showeprint[arxiv]{1803.02294}~[math.LO]
\urldef\tempurl%
\url{https://arxiv.org/abs/1803.02294}
\showURL{%
\tempurl}


\bibitem[\protect\citeauthoryear{Escard{\'{o}}}{Escard{\'{o}}}{2019}]%
        {Escardo2019}
\bibfield{author}{\bibinfo{person}{Mart{\'{\i}}n~H{\"{o}}tzel Escard{\'{o}}}.}
  \bibinfo{year}{2019}\natexlab{}.
\newblock \showarticletitle{Introduction to Univalent Foundations of
  Mathematics with Agda}.
\newblock \bibinfo{journal}{\emph{CoRR}}  \bibinfo{volume}{abs/1911.00580}
  (\bibinfo{year}{2019}).
\newblock
\showeprint[arXiv]{1911.00580}
\urldef\tempurl%
\url{http://arxiv.org/abs/1911.00580}
\showURL{%
\tempurl}


\bibitem[\protect\citeauthoryear{Gonthier}{Gonthier}{2008}]%
        {Gonthier2008}
\bibfield{author}{\bibinfo{person}{Georges Gonthier}.}
  \bibinfo{year}{2008}\natexlab{}.
\newblock \showarticletitle{Formal proof--the four-color theorem}.
\newblock \bibinfo{journal}{\emph{Notices of the AMS}} \bibinfo{volume}{55},
  \bibinfo{number}{11} (\bibinfo{year}{2008}), \bibinfo{pages}{1382--1393}.
\newblock
\urldef\tempurl%
\url{https://doi.org/10.1.1.141.714}
\showDOI{\tempurl}


\bibitem[\protect\citeauthoryear{Grayson}{Grayson}{2017}]%
        {Grayson2017}
\bibfield{author}{\bibinfo{person}{Daniel~R. Grayson}.}
  \bibinfo{year}{2017}\natexlab{}.
\newblock \bibinfo{title}{An introduction to univalent foundations for
  mathematicians}.
\newblock
\newblock
\urldef\tempurl%
\url{https://goo.gl/SEipU1}
\showURL{%
\tempurl}


\bibitem[\protect\citeauthoryear{Gross and Tucker}{Gross and Tucker}{1987}]%
        {gross}
\bibfield{author}{\bibinfo{person}{Jonathan~L Gross} {and}
  \bibinfo{person}{Thomas~W Tucker}.} \bibinfo{year}{1987}\natexlab{}.
\newblock \bibinfo{booktitle}{\emph{Topology Graph Theory}}.
\newblock \bibinfo{publisher}{Dover}, \bibinfo{address}{Ny, Usa}. 387 pages.
\newblock


\bibitem[\protect\citeauthoryear{Gross, Yellen, and Anderson}{Gross
  et~al\mbox{.}}{2018}]%
        {Gross2018}
\bibfield{author}{\bibinfo{person}{Jonathan~L. Gross}, \bibinfo{person}{Jay
  Yellen}, {and} \bibinfo{person}{Mark Anderson}.}
  \bibinfo{year}{2018}\natexlab{}.
\newblock \bibinfo{booktitle}{\emph{Graph Theory and Its Applications}}.
\newblock \bibinfo{publisher}{Chapman and Hall/{CRC}}, \bibinfo{address}{Ny,
  Usa}.
\newblock
\urldef\tempurl%
\url{https://doi.org/10.1201/9780429425134}
\showDOI{\tempurl}


\bibitem[\protect\citeauthoryear{Haar}{Haar}{2016}]%
        {Haar2014}
\bibfield{author}{\bibinfo{person}{Stefan Haar}.}
  \bibinfo{year}{2016}\natexlab{}.
\newblock \showarticletitle{Cyclic Ordering Through Partial Orders}.
\newblock \bibinfo{journal}{\emph{J. Multiple Valued Log. Soft Comput.}}
  \bibinfo{volume}{27}, \bibinfo{number}{2-3} (\bibinfo{year}{2016}),
  \bibinfo{pages}{209--228}.
\newblock
\urldef\tempurl%
\url{http://www.oldcitypublishing.com/journals/mvlsc-home/mvlsc-issue-contents/mvlsc-volume-27-number-2-3-2016/mvlsc-27-2-3-p-209-228/}
\showURL{%
\tempurl}


\bibitem[\protect\citeauthoryear{Haless, Adams, Bauer, Dang, Harrison, Hoang,
  Kaliszyk, Magron, Mclaughlin, Nguyen, and et~al.}{Haless
  et~al\mbox{.}}{2017}]%
        {hales-kepler}
\bibfield{author}{\bibinfo{person}{Thomas Haless}, \bibinfo{person}{Mark
  Adams}, \bibinfo{person}{Gertrud Bauer}, \bibinfo{person}{Tat~dat Dang},
  \bibinfo{person}{John Harrison}, \bibinfo{person}{Le~truong Hoang},
  \bibinfo{person}{Cezary Kaliszyk}, \bibinfo{person}{Victor Magron},
  \bibinfo{person}{Sean Mclaughlin}, \bibinfo{person}{Tat~Thang Nguyen}, {and}
  \bibinfo{person}{et al.}} \bibinfo{year}{2017}\natexlab{}.
\newblock \showarticletitle{A Formal Proof Of The Kepler Conjecture}.
\newblock \bibinfo{journal}{\emph{Forum of Mathematics, Pi}}
  \bibinfo{volume}{5} (\bibinfo{year}{2017}), \bibinfo{pages}{e2}.
\newblock
\urldef\tempurl%
\url{https://doi.org/10.1017/fmp.2017.1}
\showDOI{\tempurl}


\bibitem[\protect\citeauthoryear{Harrison}{Harrison}{2008}]%
        {harrison-notices}
\bibfield{author}{\bibinfo{person}{John Harrison}.}
  \bibinfo{year}{2008}\natexlab{}.
\newblock \showarticletitle{Formal Proof --Theory and Practice}.
\newblock \bibinfo{journal}{\emph{Notices of the American Mathematical
  Society}}  \bibinfo{volume}{55} (\bibinfo{year}{2008}),
  \bibinfo{pages}{1395--1406}.
\newblock


\bibitem[\protect\citeauthoryear{{Homotopy Type Theory and Univalent
  Foundations CAS project}}{{Homotopy Type Theory and Univalent Foundations CAS
  project}}{2019}]%
        {symmetrybook}
\bibfield{author}{\bibinfo{person}{The {Homotopy Type Theory and Univalent
  Foundations CAS project}}.} \bibinfo{year}{2019}\natexlab{}.
\newblock \bibinfo{title}{Symmetry Book}.
\newblock
\newblock
\urldef\tempurl%
\url{http://github.com/UniMath/SymmetryBook}
\showURL{%
\tempurl}


\bibitem[\protect\citeauthoryear{MacLane}{MacLane}{1937}]%
        {maclane}
\bibfield{author}{\bibinfo{person}{Saunders MacLane}.}
  \bibinfo{year}{1937}\natexlab{}.
\newblock \bibinfo{title}{A combinatorial condition for planar graphs}.
\newblock
\newblock


\bibitem[\protect\citeauthoryear{Mohar}{Mohar}{1988}]%
        {mohar1988}
\bibfield{author}{\bibinfo{person}{Bojan Mohar}.}
  \bibinfo{year}{1988}\natexlab{}.
\newblock \showarticletitle{Embeddings of infinite graphs}.
\newblock \bibinfo{journal}{\emph{Journal of Combinatorial Theory, Series B}}
  \bibinfo{volume}{44}, \bibinfo{number}{1} (\bibinfo{year}{1988}),
  \bibinfo{pages}{29--43}.
\newblock
\showISSN{0095-8956}
\urldef\tempurl%
\url{https://doi.org/10.1016/0095-8956(88)90094-9}
\showDOI{\tempurl}


\bibitem[\protect\citeauthoryear{Noschinski}{Noschinski}{2015}]%
        {Noschinski2015}
\bibfield{author}{\bibinfo{person}{Lars Noschinski}.}
  \bibinfo{year}{2015}\natexlab{}.
\newblock \emph{\bibinfo{title}{Formalizing Graph Theory and Planarity
  Certificates}}.
\newblock \bibinfo{thesistype}{Ph.D. Dissertation}.
  \bibinfo{school}{Technischen Universit\"{a}t M\"{u}nchen},
  \bibinfo{address}{Germany}.
\newblock
\urldef\tempurl%
\url{https://d-nb.info/1104933624/34}
\showURL{%
\tempurl}


\bibitem[\protect\citeauthoryear{Prieto-Cubides}{Prieto-Cubides}{2019}]%
        {agdaformalisation}
\bibfield{author}{\bibinfo{person}{Jonathan Prieto-Cubides}.}
  \bibinfo{year}{2019}\natexlab{}.
\newblock \bibinfo{booktitle}{\emph{Investigations on graph-theoretical
  constructions in Homotopy type theory -- Agda formalisation}}.
\newblock
\urldef\tempurl%
\url{https://doi.org/10.5281/zenodo.5775570}
\showDOI{\tempurl}
\newblock
\shownote{Work-in-progress.}


\bibitem[\protect\citeauthoryear{Prieto-Cubides}{Prieto-Cubides}{2021}]%
        {homotopywalks}
\bibfield{author}{\bibinfo{person}{Jonathan Prieto-Cubides}.}
  \bibinfo{year}{2021}\natexlab{}.
\newblock \showarticletitle{On homotopy of walks and spherical maps in homotopy
  type theory}.
\newblock  (\bibinfo{year}{2021}).
\newblock
\urldef\tempurl%
\url{https://doi.org/10.1145/3497775.3503671}
\showDOI{\tempurl}


\bibitem[\protect\citeauthoryear{Rahman}{Rahman}{2017}]%
        {Rahman2017}
\bibfield{author}{\bibinfo{person}{Md.~Saidur Rahman}.}
  \bibinfo{year}{2017}\natexlab{}.
\newblock \bibinfo{booktitle}{\emph{Planar Graphs}}.
\newblock \bibinfo{publisher}{Springer International Publishing},
  \bibinfo{address}{Cham}, \bibinfo{pages}{77--89}.
\newblock
\urldef\tempurl%
\url{https://doi.org/10.1007/978-3-319-49475-3\_6}
\showDOI{\tempurl}


\bibitem[\protect\citeauthoryear{Stahl}{Stahl}{1978}]%
        {Stahl1978}
\bibfield{author}{\bibinfo{person}{Saul Stahl}.}
  \bibinfo{year}{1978}\natexlab{}.
\newblock \showarticletitle{The embeddings of a graph--A survey}.
\newblock \bibinfo{journal}{\emph{Journal of Graph Theory}}
  \bibinfo{volume}{2}, \bibinfo{number}{4} (\bibinfo{year}{1978}),
  \bibinfo{pages}{275--298}.
\newblock
\showISSN{10970118}
\urldef\tempurl%
\url{https://doi.org/10.1002/jgt.3190020402}
\showDOI{\tempurl}


\bibitem[\protect\citeauthoryear{Team}{Team}{2021}]%
        {agda}
\bibfield{author}{\bibinfo{person}{The Agda~Development Team}.}
  \bibinfo{year}{2021}\natexlab{}.
\newblock \bibinfo{title}{Agda 2.6.1.3 documentation}.
\newblock
\newblock
\urldef\tempurl%
\url{https://agda.readthedocs.io/en/v2.6.1.3/}
\showURL{%
\tempurl}


\bibitem[\protect\citeauthoryear{{Univalent Foundations Program}}{{Univalent
  Foundations Program}}{2013}]%
        {hottbook}
\bibfield{author}{\bibinfo{person}{The {Univalent Foundations Program}}.}
  \bibinfo{year}{2013}\natexlab{}.
\newblock \bibinfo{booktitle}{\emph{Homotopy Type Theory: Univalent Foundations
  of Mathematics}}.
\newblock \bibinfo{publisher}{\url{https://homotopytypetheory.org/book}},
  \bibinfo{address}{Institute for Advanced Study}.
\newblock


\bibitem[\protect\citeauthoryear{Voevodsky}{Voevodsky}{2010}]%
        {voevodsky2014equivalence}
\bibfield{author}{\bibinfo{person}{Vladimir Voevodsky}.}
  \bibinfo{year}{2010}\natexlab{}.
\newblock \bibinfo{title}{The equivalence axiom and univalent models of type
  theory. (Talk at CMU on February 4, 2010)}.
\newblock , \bibinfo{numpages}{1{\textendash}11}~pages.
\newblock
\urldef\tempurl%
\url{https://arxiv.org/abs/1402.5556}
\showURL{%
\tempurl}


\bibitem[\protect\citeauthoryear{Whitney}{Whitney}{1932}]%
        {Whitney1932}
\bibfield{author}{\bibinfo{person}{Hassler Whitney}.}
  \bibinfo{year}{1932}\natexlab{}.
\newblock \showarticletitle{Non-separable and planar graphs}.
\newblock \bibinfo{journal}{\emph{Trans. Amer. Math. Soc.}}
  \bibinfo{volume}{34}, \bibinfo{number}{2} (\bibinfo{date}{Feb.}
  \bibinfo{year}{1932}), \bibinfo{pages}{339--339}.
\newblock
\urldef\tempurl%
\url{https://doi.org/10.1090/s0002-9947-1932-1501641-2}
\showDOI{\tempurl}


\bibitem[\protect\citeauthoryear{Yamamoto, Nishizaki, Hagiya, and
  Toda}{Yamamoto et~al\mbox{.}}{1995}]%
        {yamamoto}
\bibfield{author}{\bibinfo{person}{Mitsuharu Yamamoto},
  \bibinfo{person}{Shin{-}ya Nishizaki}, \bibinfo{person}{Masami Hagiya}, {and}
  \bibinfo{person}{Yozo Toda}.} \bibinfo{year}{1995}\natexlab{}.
\newblock \showarticletitle{Formalization of Planar Graphs}. In
  \bibinfo{booktitle}{\emph{Higher Order Logic Theorem Proving and Its
  Applications, 8th International Workshop, Aspen Grove, UT, USA, September
  11-14, 1995, Proceedings}} \emph{(\bibinfo{series}{Lecture Notes in Computer
  Science})}, \bibfield{editor}{\bibinfo{person}{E.~Thomas Schubert},
  \bibinfo{person}{Phillip~J. Windley}, {and} \bibinfo{person}{Jim
  Alves{-}Foss}} (Eds.), Vol.~\bibinfo{volume}{971}.
  \bibinfo{publisher}{Springer}, \bibinfo{address}{Ut, Usa},
  \bibinfo{pages}{369--384}.
\newblock
\urldef\tempurl%
\url{https://doi.org/10.1007/3-540-60275-5\_77}
\showDOI{\tempurl}


\bibitem[\protect\citeauthoryear{Yorgey}{Yorgey}{2014}]%
        {Yorgey}
\bibfield{author}{\bibinfo{person}{Brent~Abraham Yorgey}.}
  \bibinfo{year}{2014}\natexlab{}.
\newblock \emph{\bibinfo{title}{Combinatorial Species And Labelled
  Structures}}.
\newblock \bibinfo{thesistype}{Ph.D. Dissertation}. \bibinfo{school}{University
  of Pennsylvania}, \bibinfo{address}{Pa, Usa}.
\newblock
\urldef\tempurl%
\url{https://www.cis.upenn.edu/~sweirich/papers/yorgey-thesis.pdf}
\showURL{%
\tempurl}


\end{thebibliography}
